\documentclass[11pt]{article}

\usepackage{mathtools, fancyhdr, xcolor}
\usepackage{graphicx, amssymb, latexsym, amsmath, amsfonts, amsthm,bbm}
\mathtoolsset{showonlyrefs}

\theoremstyle{plain}
\newtheorem{theorem}{Theorem}[section]
\newtheorem{lemma}[theorem]{Lemma}
\newtheorem{claim}[theorem]{Claim}

\newtheorem{fact}[theorem]{Fact}
\newtheorem{corollary}[theorem]{Corollary}

\theoremstyle{definition}

\newtheorem{remark}[theorem]{Remark}

\newcommand{\lf}{\lfloor}
\newcommand{\rf}{\rfloor}
\newcommand{\floor}[1]{\lf {#1} \rf}

\newcommand{\ceil}[1]{\lceil {#1}\rceil}

\newcommand{\defeq}{\vcentcolon=}

\newcommand{\ind}[1]{^{(#1)}}
\newcommand\eps{\varepsilon}

\renewcommand{\Pr}{\mathop{\bf Pr\/}}
\newcommand{\E}{\mathop{\bf E\/}}

\newcommand{\Var}{\mathop{\bf Var\/}}

\DeclareMathOperator*\poly{poly}
\DeclareMathOperator*\rank{rank}
\DeclareMathOperator*\spn{span}
\newcommand\Vol{\textnormal{Vol}}

\newcommand\calE{\mathcal{E}}
\newcommand\RR{\mathbb{R}}
\newcommand\NN{\mathbb{N}}

\newcommand{\lbr}{\left\{}
\newcommand{\rbr}{\right\}}

\newcommand{\bra}[1]{\lbr#1\rbr}

\newcommand\sar{^*}

\newcommand\half {\frac{1}{2}}

\mathtoolsset{showonlyrefs}

\begin{document}
\title{Improved list-decodability of random linear binary codes\thanks{A conference version of this paper appeared in RANDOM '18.}}
\author{Ray Li\thanks{Department of Computer Science, Stanford University. This work was supported by the National Science Foundation Graduate Research Fellowship Program under Grant No. DGE - 1656518.}\ \ and Mary Wootters\thanks{Departments of Computer Science and Electrical Engineering, Stanford University. This work was partially supported by NSF grants CCF-1657049 and CCF-1844628 and by the NSF/BSF grant CCF-1814629.}}
\date{\today}
\maketitle

\begin{abstract}
There has been a great deal of work establishing that random linear codes are as list-decodable as uniformly random codes, in the sense that a random linear binary code of rate $1 - H(p) - \epsilon$ is $(p,O(1/\epsilon))$-list-decodable with high probability.
In this work, we show that such codes are $(p, H(p)/\epsilon + 2)$-list-decodable with high probability, for any $p \in (0, 1/2)$ and $\epsilon > 0$.
In addition to improving the constant in known list-size bounds, our argument---which is quite simple---works simultaneously for all values of $p$, while previous works obtaining $L = O(1/\epsilon)$ patched together different arguments to cover different parameter regimes.  

Our approach is to strengthen an existential argument of (Guruswami, H{\aa}stad, Sudan and Zuckerman, IEEE Trans. IT, 2002) to hold with high probability.  
To complement our upper bound for random linear codes, we also improve an argument of (Guruswami, Narayanan, IEEE Trans. IT, 2014) to obtain an essentially tight lower bound of $1/\epsilon$ on the list size of uniformly random codes; this implies that random linear codes are in fact \em more \em list-decodable than uniformly random codes, in the sense that the list sizes are strictly smaller.

To demonstrate the applicability of these techniques, we use them to (a) obtain more information about the distribution of list sizes of random linear codes and (b) to prove a similar result for random linear rank-metric codes.  
\end{abstract}
\thispagestyle{empty}
\newpage
\clearpage
\setcounter{page}{1}

\section{Introduction}\label{sec:intro}
An \em error correcting code \em is a subset $\mathcal{C} \subseteq \mathbb{F}_2^n$, which is ideally ``spread out."  In this paper, we focus on one notion of ``spread out" known as \em list decodability. \em  We say that a code $\mathcal{C}$ is $(p,L)$-list-decodable if any Hamming ball of radius $pn$ in $\mathbb{F}_2^n$ contains at most $L$ points of $\mathcal{C}$: that is, if for all $x \in \mathbb{F}_2^n$, $|\mathcal{B}(x,pn) \cap \mathcal{C}| \leq L$, where $\mathcal{B}(x,pn)$ is the Hamming ball of radius $pn$ centered at $x$.  Since list decoding was introduced in the 1950's~\cite{Elias57, Wozencraft58}, it has found applications beyond communication, for example in pseudorandomness~\cite{Vadhan12} and complexity theory~\cite{Sudan00}.  

A classical result in list decoding is known as the \em list-decoding capacity theorem: \em 
\begin{theorem}[List decoding capacity theorem]\label{thm:listdeccap}
  Let $p\in(0,1/2)$ and $\eps>0$. 
  \begin{enumerate}
    \item There exist binary codes of rate $1-H(p)-\eps$ that are $(p, \ceil{1/\eps})$-list decodable.
    \item Any binary code of rate $1-H(p)+\eps$ that is $(p,L)$-list decodable up to distance $p$ must have $L\ge 2^{\Omega(\eps n)}$.
  \end{enumerate}
\end{theorem}
Above, $H(p) = -(1 - p)\log_2(1-p) - p \log_2(p)$ is the binary entropy function.
We say that a family of binary codes with rate approaching $1 - H(p)$ which are $(p,L)$-list-decodable for $L = O(1)$ \em achieves list-decoding capacity.\em%
\footnote{Sometimes the phrase ``achieves list-decoding capacity" is also used when $L = \poly(n)$; since this paper focuses on the exact constant in the $O(1)$ term however, we use it to mean that $L = O(1)$.}

Theorem~\ref{thm:listdeccap} is remarkable because it 
means than even when $pn$ is much larger than half the minimum distance of the code---the radius at which at most one  codeword $c \in \mathcal{C}$ lies in any Hamming ball---it still can be the case that only a constant number of $c \in \mathcal{C}$ lie in any Hamming ball of radius $pn$. 
Because of this, there has been a great deal of work attempting to understand what codes achieve the bound in Theorem~\ref{thm:listdeccap}.  

The existential part of Theorem~\ref{thm:listdeccap} is proved by showing that a uniformly random subset of $\mathbb{F}_2^n$ is $(p,1/\eps)$-list decodable with high probability.  
For a long time, uniformly random codes were the only example of binary codes known to come close to this bound, and today we still do not have many other options.  There are explicit constructions of capacity-achieving list-decodable codes over large alphabets (either growing with $n$ or else large-but-constant)~\cite{DvirL12,GuruswamiX12, GuruswamiX13}, but over binary alphabets we still do not have any explicit constructions; we refer the reader to the survey~\cite{Guruswami09} for an overview of progress in this area. 

Because it is a major open problem to construct explicit binary codes of rate $1 - H(p) - \eps$ with constant (or even $\poly(n)$) list-sizes, one natural line of work has been to study structured random approaches, in particular \em random linear codes. \em  
A random linear code $\mathcal{C} \subset \mathbb{F}_2^n$ is simply a random subspace of $\mathbb{F}_2^n$, and the list-decodability of these codes has been well-studied~\cite{ZyablovP81, GuruswamiHSZ02, GuruswamiHK11, CheraghchiGV13, Wootters13, RudraW14, RudraW18}. 
There are several reasons to study the list-decodability of random linear codes.  Not only is it a natural question in its own right as well as a natural stepping stone in the quest to obtain explicit binary list-decodable codes, but also the list-decodabilility of random linear codes is useful in other coding-theoretic applications.  One example of this is in concatenated codes and related constructions~\cite{GI04, GuruswamiR08b, HW15, HRW17}, where a random linear code is used as a short inner code.
Here, the linearity is useful because (a) a linear code can be efficiently described; (b) it is sometimes desirable to obtain a linear code at the end of the day, hence all components of the construction must be linear; and (c) as in \cite{HW15} sometimes the linearity is required for the construction to work.

To this end, the line of work mentioned above has aimed to establish that random linear codes are ``as list-decodable'' as uniformly random codes.  That is, uniformly random codes are viewed (as is often the case in coding theory) as the optimal construction, and we try to approximate this optimality with random linear codes, despite the additional structure.  

\paragraph{Our contributions.}
In this paper, we give an improved analysis of the list-decodability of random linear binary codes.  More precisely, our contributions are as follows.
\begin{itemize}
	\item \textbf{A unified analysis.}  As we discuss more below, previous work on the list-decodability of random linear binary codes either work only in certain (non-overlapping) parameter regimes \cite{GuruswamiHK11, Wootters13}, or else get substantially sub-optimal bounds on the list-size~\cite{RudraW18}.  Our argument obtains improved list size bounds over all these results and works in all parameter regimes.  

Our approach is surprisingly simple: we adapt an existential argument of Guruswami, H{\aa}stad, Sudan and Zuckerman~\cite{GuruswamiHSZ02} to hold with high probability. Extending the argument in this way was asked as an open question in \cite{GuruswamiHSZ02} and had been open until now. 

	\item \textbf{Improved list-size for random linear codes.} 
Not only does our result imply that random linear codes of rate $1 - H(p) - \eps$ are $(p,L)$-list-decodable with list-size $L = O(1/\eps)$, in fact we show that $L \leq H(p)/\eps + 2$.  In particular, the leading constant is small and---to the best of our knowledge---is the best known, even existentially, for any list-decodable code.

  \item \textbf{Tight list-size lower bound for uniformly random codes.} To complement our upper bound, we strengthen an argument of Guruswami and Narayanan~\cite{GuruswamiN14} to show that a uniformly random binary code of rate $1 - H(p) - \eps$ requires $L \geq (1 - \gamma)/\eps$ for any constant $\gamma > 0$ and sufficiently small $\eps$. In other words, the list size of $1/\varepsilon$ in Theorem~\ref{thm:listdeccap} is tight even in the leading constant. Thus, random linear codes are, with high probability, list-decodable with smaller list sizes than completely random codes.%
\footnote{In retrospect, this may not be surprising: for example, it is well-known that random linear codes have better distance than completely random codes.  However, the fact that we are able to prove this is surprising to the authors, since it requires taking advantage of the dependence between the codewords, rather than trying to get around it.}

	\item \textbf{Finer-grained information about the combinatorial structure of random linear codes.}  We extend our argument to obtain more information about the distribution of list sizes of random linear codes.
More precisely, we obtain high-probability bounds on the number of points $x$ so that the \em list size at $x$, \em $L_{\mathcal{C}}(x) \defeq |\mathcal{B}(x,pn) \cap \mathcal{C}|$, is at least $\ell$.  

	\item \textbf{Results for rank-metric codes.}  Finally, we adapt our argument for random linear codes to apply to random linear rank-metric codes.
As with standard (Hamming-metric) codes, recent work aimed to show that
random linear rank-metric codes are nearly as list-decodable as uniformly random codes~\cite{Ding15, GuruswamiR17}.  
Our approach establishes that in fact, random linear binary rank-metric codes are more list-decodable than their uniformly random counterparts in certain parameter regimes, in the sense that the list sizes near capacity are strictly smaller.
Along the way, we show that low-rate random linear binary rank-metric codes are list-decodable to capacity, answering a question of \cite{GuruswamiR17}.

\end{itemize}

On the downside, we note that our arguments only work for binary codes and do not extend to larger alphabets; additionally, our positive results do not establish \em average-radius \em list-decodability, a stronger notion which was established in some of the previous works~\cite{CheraghchiGV13,Wootters13,RudraW18}. 
Subsequent work \cite{GLMRSW20} extends our main result on list-decoding random linear binary codes to average-radius list-decodability.

\subsection{Outline of paper}
After a brief overview of the notation in \S\ref{sec:notation}, we proceed
in \S\ref{sec:prevwork} with a survey of related work for both random linear codes and rank-metric codes, and we formally state our results in this context.
In \S\ref{sec:ld}, we prove Theorem~\ref{thm:ld}, which establishes our upper bound for random linear binary codes. 
In \S\ref{sec:listsize}, we expand upon the ideas in the proof of Theorem~\ref{thm:ld} to prove Theorem~\ref{thm:listsize}, characterizing the list size distribution of random linear codes.
In \S\ref{sec:rank-metric}, we prove Theorem~\ref{thm:rank-metric}, which adapts our upper bound to random linear binary rank-metric codes.
To round out the story, we must prove lower bounds on the list sizes of uniformly random codes for both standard and rank-metric codes.  We state these lower bounds in 
Theorems~\ref{thm:lb} and \ref{thm:rank-metric-lb}.  These proofs closely follow the approach in \cite{GuruswamiN14} and are proved in Appendices~\ref{app:lb} and \ref{app:rank-metric-lb}, respectively.

\subsection{Notation}\label{sec:notation}
We use standard Landau notation $O(\cdot), \Omega(\cdot), \Theta(\cdot)$.
We use the notation $\Omega_{x}(\cdot)$ to mean that a multiplicative factor depending on the variable(s) $x$ is suppressed.
Throughout most of the paper, we are interested in binary codes $\mathcal{C} \subseteq \mathbb{F}_2^n$ of \em block length \em $n$.  
The \em dimension \em of a code $\mathcal{C}$ is defined as $k = \log_2|\mathcal{C}|$, and the \em rate \em is the ratio $k/n$.
We say that a binary code is \em linear \em if it forms a linear subspace of $\mathbb{F}_2^n$. 
We define a \em random linear binary code of rate $R$ \em to be the span of $k = Rn$ independently random vectors $b_1,\ldots, b_k \in \mathbb{F}_2^n$.\footnote{We note that this definition is slightly different than the standard definition, which is a uniformly random subspace of dimension $k$.
However, our definition is easier to work with.
Furthermore, since uniformly random $b_1,\ldots,b_k$ have rank strictly less than $k$ with probability at most $2^{-(n-k)}$, and since each dimension $k$ code is represented in the same number of ways, our results hold also for the more standard definition.}

For a boolean statement $\varphi$, let $\mathbb{I}[\varphi]$ be 1 if $\varphi$ is true and 0 if $\varphi$ is not true.
For two points $x,y \in \mathbb{F}_2^n$, we use $\Delta(x,y) = \sum_{i=1}^n \mathbb{I}[x_i = y_i]$ to denote the Hamming distance between $x$ and $y$,
where, for an event $\calE$, $\mathbb{I}[\calE]$ is $1$ if $\calE$ occurs and $0$ otherwise. 
For $x\in \mathbb{F}_2^n, r\in[0,n]$, we define the Hamming ball $\mathcal{B}(x,r)$ of radius $r$ centered at $x$ to be $\mathcal{B}(x,r) = \{ y \in \mathbb{F}_2^n \,:\, \Delta(x,y) \leq r \}$, and the volume of $\mathcal{B}(x,r)$ to be 
$ \Vol(n, r) \defeq  |\mathcal{B}(0^n, r)| = \sum_{i=0}^r \binom{n}{i}. $
We use the well known bound that, for any $p\in[0,1/2]$, $\Vol(n,pn) \le 2^{H(p)n}$, where $H(p) = -(1 - p)\log_2(1-p) - p \log_2(p)$ is the binary entropy function.
One of our main technical results is about the distribution of \em list sizes \em of points $x \in \mathbb{F}_2^n$: given a code $\mathcal{C}$ and $p \in (0,1/2)$, 
we define the list size of a point $x\in\mathbb{F}_2^n$ to be $L_\mathcal{C}(x)\defeq|\mathcal{B}(x,pn)\cap\mathcal{C}|$. 

For $\alpha > 0, \beta\in\RR$, let $\exp_\alpha(\beta) \defeq \alpha^\beta$ and assume $\alpha =e$ when it is omitted.
For two sets $A,B\subseteq \mathbb{F}_2^n$, define the sumset $A+B = \{a+b:a\in A, b\in B\}$. When $b\in\mathbb{F}_2^n$, let $A+b$ denote $A+\{b\}$.

\section{Previous Work and Our Results}\label{sec:prevwork}
In \S\ref{ssec:rlc} below, we survey related work on the list-decodability of random linear binary codes, and state our results.  In \S\ref{ssec:rankmetric}, we do the same for our results on the list-decodability of random linear binary rank-metric codes.

\subsection{Random Linear Codes}\label{ssec:rlc}
The list-decodability of random linear binary codes has been well studied.
Here we survey the results that are most relevant for this work.  
As this work focuses on binary codes, we focus this survey on results for binary codes, even though many of the works mentioned also apply to general $q$-ary codes.
We additionally remark that, in contrast to the large alphabet setting \cite{GuruswamiR08}, capacity achieving binary codes have no known explicit constructions.

A  modification of the proof of the list decoding capacity theorem shows that 
a random linear code of rate $1-H(p)-\eps$ is $(p, \exp(O(\frac{1}{\eps})))$-list decodable~\cite{ZyablovP81}.
However, whether or not random linear codes of this rate with list-sizes that do not depend exponentially on $\eps$ remained open for decades: this question was explicitly asked in~\cite{Elias91}.  

A first step was given in the work of Guruswami, H{\aa}stad, Sudan and Zuckerman~\cite{GuruswamiHSZ02}, who proved via a beautiful potential-function-based-argument that there \em exist \em binary linear codes or rate $1 - H(p) - \eps$ which are $(p, 1/\eps)$-list-decodable.  However, this result did not hold with high probability.  Our approach relies heavily on the approach of \cite{GuruswamiHSZ02}, and we return to their argument in \S\ref{sec:ld}.

Over the next 15 years, a line of work~\cite{GuruswamiHK11, CheraghchiGV13, Wootters13, RudraW14, RudraW15, RudraW18} has focused on the list-decodability (and related properties) of random linear codes, which should hold with high probability.  The works most relevant to ours are \cite{GuruswamiHK11,Wootters13}, which together more or less settle the question. We state these results here for binary alphabets, although both works address larger alphabets as well. 

The first result, of \cite{GuruswamiHK11}, establishes a result for a constant $p$, bounded away from $1/2$.
\begin{theorem}[Theorem 2 of \cite{GuruswamiHK11}]
  \label{thm:ghk11}
  Let $p\in(0,1/2)$. Then there exist constants $C_p,\delta>0$ such that for all $\eps>0$ and sufficiently large $n$, for all $R\leq1-H(p)-\eps$, if $\mathcal{C}\subseteq\mathbb{F}_2^n$ is a random linear code of rate $R$, then $\mathcal{C}$ is $(p,C_p/\eps)$-list decodable with probability at least $1 - 2^{-\delta n}$.
\end{theorem}
However, $C_p$ is not small and tends to $\infty$ as $p$ approaches $1/2$.  
The following result of~\cite{Wootters13} fills in the gap when $p$ is quite close to $1/2$.
\begin{theorem}[Theorem 2 of \cite{Wootters13}]
  \label{thm:wootters13}
There exist constants $C_1, C_2$ so that for all sufficiently small $\eps > 0$ and sufficiently large $n$, for $p = 1/2 - C_1\sqrt{\eps}$ and for all $R \leq 1 - H(p) - \eps$, if $\mathcal{C} \subseteq \mathbb{F}_2^n$ is a random linear code of rate $R$, then 
$\mathcal{C}$ is $(p,C_2/\eps)$-list decodable with probability at least $1 - o(1)$.
\end{theorem}

The list-decoding capacity theorem implies that we cannot hope to take the rate $R$ substantially larger than $1 - H(p) - \eps$ and obtain a constant list size.  Moreover, the list size $\Theta(1/\eps)$ is optimal for both random linear codes and uniformly random codes~\cite{Rudra11,GuruswamiN14}.  More precisely,
Guruswami and Narayanan show the following theorem (which we have specialized to binary codes).

\begin{theorem}[Theorem 20 of \cite{GuruswamiN14}]
Let $\eps > 0$.
  A uniformly random binary code of rate $1-H(p)-\eps$ is $(p,(1-H(p))/\eps)$-list decodable with probability at most $\exp(-\Omega_{p,\eps}(n))$.\footnote{In fact, in \cite{GuruswamiN14}, Theorem 20 is stated with a list size of $(1 - H(p))/2\varepsilon$ as a lower bound.  However, the constant can be improved to $1 - H(p)$, because the factor of $2$ is introduced to handle an additive constant term.  Thus, for sufficiently large $n$ their argument proves the stronger statement stated above.}
\end{theorem}

We note that for general codes (not uniformly random or random linear) it is still unknown what the ``correct" list size $L$ is in terms of $\eps$, although there are results in particular parameter regimes~\cite{Blinovsky86,GuruswamiV05} and for stronger notions of list-decodability~\cite{GuruswamiN14}.

\subsubsection{Our results for random linear codes}
We show that high probability a random linear binary code of rate $1 - H(p) - \eps$ is $(p,L)$-list-decodable with $L \sim H(p)/\eps$, while a uniformly random binary code of the same rate requires $L \geq (1 - \gamma)/\eps$.
More precisely, the upper bound is as follows (proved in \S\ref{sec:ld}).
\newcommand{\thmld}{
  Let $\eps>0$ and $p\in(0,1/2)$.
  A random linear code of rate $1-H(p)-\varepsilon$ is $(p,H(p)/\eps+2)$-list decodable with probability $1-\exp(-\Omega_{\eps}(n))$.
}
\begin{theorem}
  \label{thm:ld} % ld for ``list decoding''
  \thmld
\end{theorem}
Theorem~\ref{thm:ld} improves upon the picture given by Theorems~\ref{thm:ghk11} and \ref{thm:wootters13} in two ways.
First, the leading constant on the list size, which is $H(p)$, improves over both the constant $C_p$ from Theorem~\ref{thm:ghk11} (which blows up as $p \to 1/2$) and on the constant $C_2$ from Theorem~\ref{thm:wootters13} (which the authors do not see how to make less than $2$).  Moreover, when $p \to 1/2$, Theorem~\ref{thm:ld} improves on Theorem~\ref{thm:wootters13} in that it decouples $p$ from $\eps$: in Theorem~\ref{thm:wootters13}, we must take $p = 1/2 - O(\sqrt{\eps})$ and $R = 1 - H(p) - \eps$, while in Theorem~\ref{thm:ld}, $p$ and $\eps$ may be chosen independently.  Thus, Theorem~\ref{thm:ld} offers the first true ``list-decoding capacity theorem for binary linear codes," in that it precisely mirrors the quantifiers in Theorem~\ref{thm:listdeccap}.

The list size of $H(p)/\eps + 2$ is smaller than the list size of $1/\eps$ given by the classical list decoding capacity theorem for uniformly random codes.   Further, the following negative result shows that the list size of $1/\eps$ given by uniformly random binary codes in the list decoding capacity theorem is tight, even in the leading constant of 1.
\newcommand{\thmlb}{
For any $p \in (0,1/2)$ and $\varepsilon>0$, there exists a $\gamma_{p,\varepsilon}=\exp(-\Omega_p(\frac{1}{\varepsilon}))$ and $n_{p,\varepsilon}\in\NN$ such that for all $n \ge n_{p,\varepsilon}$, a random code $\mathcal{C} \subseteq \mathbb{F}_2^n$ of rate $R=1-H(p)-\eps$ is with probability $1-\exp(-\Omega_{p,\eps}(n))$ \emph{not} $(p,\frac{1-\gamma_{p,\varepsilon}}{\eps}-1)$-list decodable.
}
\begin{theorem}
  \label{thm:lb} % lb for ``lower bound''
  \thmlb
\end{theorem}
The proof of Theorem~\ref{thm:lb} is obtained by tightening the second moment method proof of \cite{GuruswamiN14}, and is given in Appendix~\ref{app:lb}.
Theorem~\ref{thm:lb}, combined with Theorem~\ref{thm:ld}, implies that, for all $p\in (0,1/2)$ and sufficiently small $\eps$, random linear codes of rate $1-H(p)-\eps$ with high probability can be list decoded up to distance $p$ with smaller list sizes than uniformly random codes.
Perhaps surprisingly, the difference between the list size upper bound in Theorem~\ref{thm:listdeccap} and the lower bound in Theorem~\ref{thm:lb} is bounded by 2 as $\varepsilon\to 0$, implying that the ``correct'' list size of a uniformly random code is tightly concentrated between $\floor{1/\varepsilon}\pm 1$ for small $\varepsilon$.

We are unaware of results in the literature that give even the existence of binary codes list decodable with list size \em better \em than $H(p)/\varepsilon$.
We remark that the Lovasz Local Lemma also gives the \emph{existence} of $(p,H(p)/\varepsilon)$ list decodable codes, matching our high probability result for random linear codes.
Though we suspect this argument is known, we are not aware of a published version, so we include a proof in Appendix~\ref{app:lll} for completeness.

The techniques that we use to prove Theorem~\ref{thm:ld} can be refined to give more combinatorial information about random linear codes.  It is our hope that such information will help in further derandomizing constructions of binary codes approaching list-decoding capacity.
In \S\ref{sec:listsize}, we prove the following structural result about random linear binary codes.
\begin{theorem}
  \label{thm:listsize}
  Fix $L\ge 1$ and $\gamma \in (0,1)$.
  For all sufficiently large $n$, if $\mathcal{C} \subseteq \mathbb{F}_2^n$ is a random linear code of dimension $k$, then with probability $1-\exp(-\Omega(\gamma n))$, for all $1\le \ell\le L$,
  \begin{align}
    \frac{|\{x:L_\mathcal{C}(x)\ge\ell\}|}{2^n}
    \ \le \ \left(2^{-n(1-H(p))}\cdot 2^k\right)^\ell \cdot 2^{\gamma\ell^2n}.
  \end{align}
\end{theorem}
To interpret this result, it is helpful to think of $\gamma$ close to 0 and $k=n(1-H(p)-\varepsilon)$.
Then, with high probability over the choice of a random linear code $\mathcal{C}$, the fraction of centers $x\in \mathbb{F}_2^n$ with ``list size at least $\ell$'' decays approximately exponentially as $2^{-n\ell\eps}$.
By an appropriately small choice of $\gamma$, Theorem~\ref{thm:ld} follows as a corollary of Theorem~\ref{thm:listsize} (see Corollary~\ref{cor:listsize-ld}), but as a warm-up to Theorem~\ref{thm:listsize} we present a proof of Theorem~\ref{thm:ld} independent of Theorem~\ref{thm:listsize} in \S\ref{sec:ld}.

\begin{remark}
Theorem~\ref{thm:listsize} implies that in a random linear code of rate $1-H(p)-\varepsilon$, with high probability over the choice of code, we have
$  
\Pr_{x}\left[ L_\mathcal{C}(x) \ge 2  \right] \lesssim 2^{-2n\varepsilon}.
$
On the other hand, we know that $\sum_{x}^{} L_C(x) = |C|\cdot B(0,pn)$, so if the maximum list size $L$ is a constant, we have $\Pr_x[L_\mathcal{C}(x)\ge 1]\ge\frac{1}{L}\cdot |\mathcal{C}|\cdot |B(0,pn)| \ge 2^{-n\ell\varepsilon(1+o(1))}$.
Thus, ``most" centers $x$ within $pn$ of a codeword $c \in \mathcal{C}$ are not within $pn$ of any other codeword $c' \neq c$. 
This is in line with the conventional wisdom from the Shannon model: with high probability, random linear codes achieve capacity on the BSC, so for a random linear code, a random center $x$ obtained by sending a codeword $c \in \mathcal{C}$ through the binary symmetric channel BSC($p$) has $L_\mathcal{C}(x)\le 1$ with high probability.
(See also~\cite{RudraU10}).  Thus, Theorem~\ref{thm:listsize} recovers this intuition for list size $1$, and quantitatively extends it to list sizes larger than $1$.
\end{remark}

\subsection{Rank-Metric Codes}\label{ssec:rankmetric}
As an application of our techniques for random linear codes, we turn our attention to \em rank metric codes. \em
Rank metric codes, introduced by Delsarte in~\cite{Delsarte78}, are codes $\mathcal{C} \subseteq \mathbb{F}_q^{m \times n}$; that is, the codewords are $m \times n$ matrices, where typically $m \geq n$.  The distance between two codewords $X$ and $Y$ is given by the rank of their difference:
$ \Delta_R(X,Y) \defeq \frac{1}{n}\rank(X- Y),$
where $\Delta_R$ is called the \emph{rank metric}.  We denote the \emph{rank ball} by 
$\mathcal{B}_{q,R}(X,p) \ \defeq \ \bra{Y \in \mathbb{F}_q^{m \times n} \,:\, \Delta_R(X-Y) \leq pn},$
and say that a rank metric code $\mathcal{C} \subseteq \mathbb{F}_q^{m \times n}$ is $(p,L)$-list-decodable if $|\mathcal{B}_{q,R}(X,p) \cap \mathcal{C}| \leq L$ for all $X \in \mathbb{F}_q^{m \times n}$.  The \em rate \em $R$ of a rank metric code $\mathcal{C} \subset \mathbb{F}_q^{m \times n}$ is $R \defeq \log_q(|\mathcal{C}|)/(mn)$.

Rank metric codes generalize standard (Hamming metric) codes, which are simply diagonal rank metric codes.  
The study of rank metric codes has been motivated by a wide range of applications, including magnetic storage~\cite{Roth91}, cryptography~\cite{GabidulinPT91, Loidreau10, Loidreau17}, space-time coding~\cite{LusinaGB03,Luk05}, network coding~\cite{KoetterK08, SilvaKK08}, and distributed storage~\cite{SilbersteinRV12, SilbersteinRKV13}.

The natural ``list-decoding capacity" for rank metric codes is $R = (1 - p)(1 - (n/m)p)$, which is the analog of the Gilbert-Varshamov bound~\cite{GadouleauY08}.  It was shown in \cite{Ding15,GuruswamiR17} that this is achievable by a uniformly random rank metric code.

\begin{theorem}[\cite{GuruswamiR17}, Proposition A.1.\footnote{
In \cite{GuruswamiR17}, the result is stated with $L = O(1/\eps)$, but an inspection of the proof shows that we may take the leading constant to be $1$.}]
\label{thm:rank-metric-0}
Let $\eps > 0$ and $p \in (0,1)$ and suppose that $m,n$ are sufficiently large compared to $1/\eps$.  A uniformly random code $\mathcal{C} \subseteq \mathbb{F}_q^{m \times n}$ 
of rate $R = (1 - p)(1 - bp) - \eps$ 
is $(p, \ceil{1/\eps})$-list-decodable with probability at least $1 - O(q^{-\eps m n })$, where $b = n/m$.
\end{theorem}
As with standard (Hamming-metric) codes, it is interesting to study the list-decodability of random linear rank-metric codes; we say that $\mathcal{C} \subseteq \mathbb{F}_q^{n \times m}$ is linear if it forms an $\mathbb{F}_q$-linear space.
It is shown in \cite{Ding15} that no linear code can beat the bound in Theorem~\ref{thm:rank-metric-0}.

\begin{theorem}
\label{thm:rank-metric-lb-1} 
Let $b = \lim_{n \to \infty} \frac{n}{m}$ be a constant and $L\le 2^{o(mn)}$.  Then for any $R \in (0,1)$ and $p \in (0,1)$, a $(p, L)$-list-decodable \emph{linear} rank metric code $\mathcal{C} \subseteq \mathbb{F}_q^{m \times n}$ with rate $R$ satisfies $R \leq (1 - p)(1 - bp)$.
\end{theorem}

There has been a great deal of work aimed at establishing (or refuting) the list-decodability of explicit rank metric codes.  It is shown in~\cite{WachterZeh13} that Gabidulin codes~\cite{Gabidulin85}---the rank-metric analog of Reed-Solomon codes---are \em not \em list-decodable to capacity, or even much beyond half their minimum distance.  However, there have been works~\cite{GuruswamiX13, GuruswamiWX16} designing explicit codes meeting the limitation of Theorem~\ref{thm:rank-metric-lb-1}.  

Following the approach of~\cite{ZyablovP81} for Hamming metric codes, \cite{Ding15} shows that random linear rank metric codes of rate $R = (1 - p)(1 - bp) - \eps$ are $(p, \exp(O(1/\eps))$-list-decodable, where as above $b = n/m$.
In a recent paper of Guruswami and Resch~\cite{GuruswamiR17}, this result was strengthened to give a list size of $O(1/\eps)$.
\begin{theorem}[\cite{GuruswamiR17}]
  \label{thm:rank-metric-2}
  Let $p\in(0,1)$ and $q\geq2$.  There is some constant $C_{p,q}$ so that the following holds.
	For all sufficiently large $n,m$ with $b = n/m$,  
a random linear rank metric code $\mathcal{C}\subseteq \mathbb{F}_q^{m \times n}$ of rate $R=(1-p)(1-bp)-\eps$ is $(p,C_{p,q}/\eps)$-list decodable with high probability.
\end{theorem}
The proof of Theorem~\ref{thm:rank-metric-2} uses ideas from the approach of \cite{GuruswamiHK11} to prove Theorem~\ref{thm:ghk11}.  This is a beautiful argument, but as with the results of \cite{GuruswamiHK11}, the result of \cite{GuruswamiR17} suffers from the fact that $C_{p,q}$ tends to $\infty$ as $p$ approaches 1.  
It is shown in \cite{GuruswamiR17}, Proposition A.2, that when $p = 1 - \eta$, a uniformly random rank metric code of rate $R = (\eta - \eta b + \eta^2 b)/2$ is $(p, 4/(\eta - \eta b + \eta^2 b))$-list-decodable, and that work poses the question of whether or not a random linear rank metric code can achieve this.  Our results, described in the next section, show that the answer is ``yes" for $q=2$.

\subsubsection{Our Results for Rank Metric Codes}
By applying the techniques in the proof of Theorem~\ref{thm:ld}, we prove the following upper bound on the list size of random linear binary rank-metric codes.
\newcommand{\thmrankmetric}{
  Let $p\in (0,1)$ and choose $\eps > 0$.  There is a constant $C_{\eps}$ so that the following holds.
  Let $m$ and $n$ be sufficiently large positive integers with $n < m$ and let $b=n/m$.
  A random linear rank metric code $\mathcal{C}\subseteq \mathbb{F}_2^{m \times n}$ of rate $R=(1-p)(1-bp)-\eps$ is $(p,\frac{p+bp-bp^2}{\eps}+2)$-list decodable with probability at least $1-\exp(-C_{\eps}mn)$.
}
\begin{theorem}
  \label{thm:rank-metric}
  \thmrankmetric
\end{theorem}
Notice that Theorem~\ref{thm:rank-metric} works for all $p$, improving upon Theorem~\ref{thm:rank-metric-2}.  
In particular, when $p = 1 - \eta$, then setting $\eps = (1 - p)(1 - bp)/2$ and applying Theorem~\ref{thm:rank-metric} implies that a random linear binary rank metric code of rate $R = (\eta - \eta b + \eta^2 b )/2$ is $(p, L)$ list-decodable for $L \leq \frac{2}{\eta - \eta b + \eta^2 b}$, answering the aforementioned open question of \cite{GuruswamiR17} in the affirmative.

We also prove a new lower bound on the list size of uniformly random rank-metric codes.
\newcommand{\thmrankmetriclbq}{
Let $\eps > 0$ and suppose $m,n$ are sufficiently large so that $b = n/m$.
  Let $\mathcal{C}\subseteq \mathbb{F}_q^{m\times n}$ be a uniformly random rank metric code of rate $R=(1-p)(1-bp)-\eps$.  Then $\mathcal{C}$ is $(p, (1-p)(1-bp)/\eps - 2)$-list decodable with probability at most $\exp(-\Omega_{p,\eps}(n))$.
}
\begin{theorem}
  \label{thm:rank-metric-lb}
  \thmrankmetriclbq
\end{theorem}
Theorem~\ref{thm:rank-metric-lb} again uses the method of \cite{GuruswamiN14}, and we prove it in Appendix~\ref{app:rank-metric-lb}.
Together,
Theorems~\ref{thm:rank-metric} and \ref{thm:rank-metric-lb} 
show that for some values of $p$, random linear binary rank metric codes have a strictly smaller list size than uniformly random rank metric codes with the same parameters.  
In particular, the upper bound of Theorem~\ref{thm:rank-metric} is strictly smaller than the lower bound of Theorem~\ref{thm:rank-metric-lb} whenever
$p < \frac{1-b}{2}$.
For larger values of $p$, we remark that the list size obtained by Theorem~\ref{thm:rank-metric} is still strictly smaller than the $1/\eps$ list size given by uniformly random codes in Theorem~\ref{thm:rank-metric-0}, even though in this case we don't have a lower bound which proves that this is tight.

%% Derivation for reference:
%We need
%\begin{align}
%  \frac{(1-p)(1-bp)}{\varepsilon} > \frac{p+bp-bp^2}{\varepsilon} = \frac{1 - (1-p)(1-bp)}{\varepsilon}
%  \iff
%  \frac{(1-p)(1-bp)}{\varepsilon} > \half.
%\label{}
%\end{align}
%  Since $b$ is small, this is saying that $p$ must be less than $1/2$ (it's actually $1/2$- a little).
%  Solving the quadratic gives
%  \begin{align}
%    p < \frac{b+1-\sqrt{b^2+1}}{2b} = \frac{1}{b+1+\sqrt{b^2+1}} \approx \half - \frac{3b}{8}
%  \label{}
%  \end{align}
%  when $b$ is small (I think it needs to be, by \cite{Ding15}. They say that, for some $c$, if $b=n/m\ge c\varepsilon$, then code cannot be list-decodable with polynomially bounded lists. \cite{GuruswamiR17} also makes this point.)
%  So we're not actually better for most values of $p$, but small values of $p$.

\section{Simplified result for random linear binary codes}\label{sec:ld}
In this section, we prove Theorem~\ref{thm:ld}, which we restate here.
\begin{theorem}[Theorem~\ref{thm:ld}, restated]
  \thmld
\end{theorem}

Theorem~\ref{thm:ld} also follows from our more refined result, Theorem~\ref{thm:listsize}.  
However, we present a different proof of Theorem~\ref{thm:ld} first, for two reasons: (1) the results in Section~\ref{sec:rank-metric} follow the argument structure in this section, and (2) the argument in this section shows how to extend the argument in \cite{GuruswamiHSZ02} to a high probability result.
Before it was known that a typical random linear code is $(p,O(1/\varepsilon))$-list decodable, Guruswami, H{\aa}stad, Sudan and Zuckerman \cite{GuruswamiHSZ02} proved the \emph{existence} of binary linear codes of rate $1-H(p) - \eps$ that are $(p,1/\eps)$-list decodable.
However, their argument did not work with high probability, and the authors explicitly stated this as a drawback of their proof.
This section shows how to make the argument in \cite{GuruswamiHSZ02} work with high probability.
We start by reviewing the approach of \cite{GuruswamiHSZ02}, which is the basis of our proof.

\subsection{The approach of \cite{GuruswamiHSZ02}}\label{sec:review}
The approach of \cite{GuruswamiHSZ02} followed from a beautiful potential-function argument, which is the basis of our approach and which we describe here.

\newcommand{\fakeS}{\tilde{S}}
\newcommand{\fakeT}{\tilde{T}}
Let $k \defeq Rn  =(1-H(p)-\eps)n$. We choose vectors $b_1,\ldots,b_k$ one at a time, so that 
the code $\mathcal{C}_i \defeq \spn(b_1,\ldots,b_i)$ remains ``nice": formally, so that a potential function $\fakeS_{\mathcal{C}_i}$ remains small.
Once we have picked all $k$ vectors, we set $\mathcal{C} = \mathcal{C}_k$, and the fact that $\fakeS_{\mathcal{C}_k}$ is small implies list-decodability.

Recall that for a code $\mathcal{C}$ and $x\in \mathbb{F}_2^n$, we set $L_\mathcal{C}(x) = |\mathcal{B}(x,pn)\cap\mathcal{C}|$.
  Define
\[ \fakeS_{\mathcal{C}} \defeq \frac{1}{2^n} \sum_{x \in \mathbb{F}_2^n} 2^{\eps n L_\mathcal{C}(x)}. \]
It is not hard to show that for any vectors $b_1,\ldots,b_i\in \mathbb{F}_2^n$,
\begin{equation}
  \label{eq:r1}
 \E_{b_{i+1} \sim \mathbb{F}_2^n} \left[ \fakeS_{\mathcal{C}_i + \{0,b_{i+1}\}} \vert b_1,\ldots,b_i \right] \le \fakeS_{\mathcal{C}_i}^2.
\end{equation}
That is, when a uniformly random vector $b_{i+1}$ is added to the basis $\{b_1,\ldots,b_{i}\}$, we expect the potential function not to grow too much.
Hence, there exists a choice of vectors $b_1,\dots, b_k$ so that $\fakeS_{\mathcal{C}_{i+1}} \leq \fakeS_{\mathcal{C}_i}^2$ for $i=0,1,\dots,k-1$.%
\footnote{As a technical detail, one needs to be careful that $b_{i+1}\notin \mathcal{C}_i$.
One can guarantee $b_{i+1}\notin\mathcal{C}_i$ by carefully examining the proof of \eqref{eq:r1}, or use \eqref{eq:r1} to get a similar equation where we additionally condition $b_{i+1}\notin \mathcal{C}_i$.
}

As $\mathcal{C}_0 = \{0\}$, we have $\fakeS_{\mathcal{C}_0} \leq 1 + 2^{-n(1 - H(p) - \eps)}$.
Setting $\mathcal{C} = \mathcal{C}_k = \spn(b_1,\ldots,b_k)$, we have
\begin{align*}
 \fakeS_{\mathcal{C}} &\leq \fakeS_{\mathcal{C}_0}^{2^k} 
\leq \left(1 + 2^{-n(1 - H(p) - \eps)}\right)^{2^k} 
\leq \exp\left( 2^{k - n(1 - H(p) - \eps)} \right) 
 = \ e 
\end{align*}
by our choice of $k$.   This implies that
 $\sum_x 2^{\eps n L_\mathcal{C}(x)} \leq e \cdot 2^n,$ 
and in particular, for all $x \in \mathbb{F}_2^n$, we have
$2^{\eps n L_\mathcal{C}(x)} \leq e \cdot 2^n$.
Thus, for all $x$, $L_\mathcal{C}(x) \leq \frac{1}{\eps} + o(1)$, as desired.

The approach of \cite{GuruswamiHSZ02} is extremely clever, but these ideas have not, to the best of our knowledge, been used in subsequent work on the list-decodability of random linear codes.  One reason is that the crux of the argument, which is \eqref{eq:r1}, holds in expectation, and it was not clear how to show that it holds with high probability; thus, the result remained existential, and other techniques were introduced to study typical random codes~\cite{GuruswamiHK11,CheraghchiGV13,Wootters13,RudraW14,RudraW18}. 

\subsection{Proof of Theorem~\ref{thm:ld}}
We improve the argument of \cite{GuruswamiHSZ02} in two ways.  First, we show that in fact, \eqref{eq:r1}
essentially holds with high probability over the choice of $b_{i+1}$, which allows us to use the approach sketched above for random linear codes.
Second, we introduce one additional trick which takes advantage of the linearity of the code in order to reduce the constant in the list size from $1$ to $H(p)$.
Before diving into the details, we briefly describe the main ideas.

The first improvement follows from looking at the potential function in the right way.
In this paragraph, all $o(1)$ terms are exponentially small in $n$.
Our goal is $\fakeS_{\mathcal{C}_k}\le O(1)$.
Write $\fakeS_{\mathcal{C}_i} = 1 + \fakeT_{\mathcal{C}_i}$.
By above, $\fakeT_{\mathcal{C}_0} = \fakeS_{\mathcal{C}_0} - 1 = o(1)$. 
We show that with high probability, for all $i\le k$, we have $\fakeT_{\mathcal{C}_i} = o(1)$. 
In the \cite{GuruswamiHSZ02} argument we have 
\[\E \fakeS_{\mathcal{C}_{i+1}} \leq \fakeS_{\mathcal{C}_i}^2 = (1 + \fakeT_{\mathcal{C}_i})^2 = 1 + 2\fakeT_{\mathcal{C}_i}(1+o(1)), \]
and so $\E \fakeT_{\mathcal{C}_{i+1}} = 2\fakeT_{\mathcal{C}_i}(1+o(1))$.
One can show that, always, $2\fakeT_{\mathcal{C}_i}\le \fakeT_{\mathcal{C}_{i+1}}$.
By Markov's inequality, $\fakeT_{\mathcal{C}_{i+1}} = 2\fakeT_{\mathcal{C}_i}(1+o(1))$ with probability $1-o(1)$, for appropriately chosen $o(1)$ terms. 
Union bounding over the $o(1)$ failure probabilities in the $k$ steps, we conclude that $\fakeT_{\mathcal{C}_i}$ grows roughly as slowly as in the existential argument, giving the desired list decodability.

The second improvement follows from the linearity of the code. In the last step of the \cite{GuruswamiHSZ02} argument, we replace the summation ``$\sum_x$" in $\sum_x 2^{\eps n L_\mathcal{C}(x)} \leq e \cdot 2^n$ with a ``$\forall x$." 
We can save a bit because, by linearity, the contribution $2^{\eps n L_\mathcal{C}(x)}$ is the same for all $x$ in a coset $y + \mathcal{C}$.

Now we go through the details.  
It is convenient to change the definition of the potential function very slightly: losing the tilde, define, for a code $\mathcal{C} \subset \mathbb{F}_2^n$,
\begin{align*}
  A_\mathcal{C}(x) \ &\defeq \  2^{\frac{\eps n L_\mathcal{C}(x)}{1+\eps}} 
  \qquad \text{and} \qquad
  S_\mathcal{C} \ \defeq \    \E_{x\sim\mathbb{F}_2^n} \left[ A_\mathcal{C}(x) \right]
  \qquad \text{and} \qquad
 T_{\mathcal{C}} \ \defeq \ S_\mathcal{C} - 1 .
\end{align*}
As noted above, it is helpful to think of $T_\mathcal{C}$ as a very small term; we would like to show---in accordance with \eqref{eq:r1}---that $T_\mathcal{C}$ approximately doubles each time we add a basis vector.
The term $S_\mathcal{C}$ differs from the term $\fakeS_\mathcal{C}$ above in that $A_\mathcal{C}(x)$ has an extra factor of $\frac{1}{1 + \eps}$ in the exponent.  This is an extra ``slack" term that helps guarantee a high probability result under the same parameters.
However, this definition does not change how the potential function behaves.  In particular, we still have the following lemma:
  \begin{lemma}[Following \cite{GuruswamiHSZ02}]\label{lem:AL-bound}
    For all linear $\mathcal{C} \subseteq \mathbb{F}_2^n$ and all $b \in \mathbb{F}_2^n$,
    \begin{align}
      L_{\mathcal{C}+\{0,b\}}(x) \ &\le \   L_\mathcal{C}(x) + L_\mathcal{C}(x+b) \label{eq:q-bound-1} \\
      A_{\mathcal{C}+\{0,b\}}(x) \ &\le \   A_\mathcal{C}(x) \cdot A_\mathcal{C}(x+b),
      \label{eq:q-bound-2}
    \end{align}
    with equality if and only if $b\notin\mathcal{C}$.
  \end{lemma}
\begin{proof}
To see \eqref{eq:q-bound-1}, notice that
    \begin{align*}
      L_{\mathcal{C}+\{0,b\}}(x) 
      \ &= \   \left|\mathcal{B}(x,pn)\cap\left( \mathcal{C}\cup (\mathcal{C}+b) \right)\right| \\
      \ &\le \  |\mathcal{B}(x,pn)\cap \mathcal{C}| + |\mathcal{B}(x,pn)\cap(\mathcal{C}+ b)| \\
      \ &= \  |\mathcal{B}(x,pn)\cap \mathcal{C}| + |\mathcal{B}(x+b,pn)\cap\mathcal{C}| \\
      \ &= \ L_{\mathcal{C}}(x) + L_\mathcal{C}(x+b),
    \end{align*}
with equality in the second line if and only if $b \not\in \mathcal{C}$.
Inequality \eqref{eq:q-bound-2} follows as a consequence of \eqref{eq:q-bound-1}, and this proves the lemma.
\end{proof}

The following lemma is the key step of our proof, establishing that when $T_\mathcal{C}$ is small, it essentially doubles with high probability each time we add a basis vector.

  \begin{lemma}
    \label{lem:ld-t}
    If $\mathcal{C}$ is a fixed linear code,
\[
      \Pr_{b\sim\mathbb{F}_2^n} \left[ S_{\mathcal{C} + \{0,b\}} \ge 1 + 2T_\mathcal{C} + T_\mathcal{C}^{1.5} \right] \ \le \   T_\mathcal{C}^{0.5}.
\]
  \end{lemma}
  \begin{proof}
    By Lemma~\ref{lem:AL-bound}, for all $b$,
    \begin{align*}
      S_{\mathcal{C} + \{0,b\}} 
      \ &= \ \E_x \left[ A_{\mathcal{C} + \{0,b\}}(x) \right] \\
      \ &\le \   \E_x\left[A_\mathcal{C}(x) A_\mathcal{C}(x+b)\right] \\
      \ &= \   \E_x\left[ -1 + A_\mathcal{C}(x) + A_\mathcal{C}(x+b) + (A_\mathcal{C}(x)-1)(A_\mathcal{C}(x+b)-1)\right] \\
      \ &= \   1+2T_{\mathcal{C}} + \E_x\left[ (A_\mathcal{C}(x)-1)(A_\mathcal{C}(x+b)-1)\right].
    \end{align*}
    Over the randomness of $b$ and $x$, we have $x$ and $x+b$ are independently uniform over $\mathbb{F}_2^n$, so
  \begin{align}
    \label{eq:ld-b-1}
      \E_b\E_x\left[(A_\mathcal{C}(x)-1) (A_{\mathcal{C}}(x+b)-1)\right] = \E_{b,x}\left[ A_\mathcal{C}(x)-1 \right] \cdot \E_{b,x}\left[A_\mathcal{C}(x+b)-1\right] = T_{\mathcal{C}}^2.
  \end{align}
    As $A_\mathcal{C}(x)-1$ is always nonnegative, we have, by Markov's inequality,
    \begin{align*}
      \Pr_b \left[ S_{\mathcal{C} + \{0,b\}} \ge 1+2T_{\mathcal{C}}+T_{\mathcal{C}}^{1.5} \right]
      \ &\le \   \Pr_b\left[\E_x\left[(A_\mathcal{C}(x)-1) (A_{\mathcal{C}}(x+b)-1)\right] \ge T_{\mathcal{C}}^{1.5}\right] 
      \ \le \ \frac{T_{\mathcal{C}}^2}{T_\mathcal{C}^{1.5}} 
      \ = \ T_\mathcal{C}^{0.5}.
      \qedhere 
    \end{align*}
  \end{proof}

Iterating Lemma~\ref{lem:ld-t} gives the following.
\begin{lemma}
  \label{lem:tech}
  Let $p\in(0,1/2)$ and $\eps\in(0,1-H(p))$.
  Let $\mathcal{C} \subset \mathbb{F}_2^n$ be a random linear code of rate $1-H(p)-\varepsilon$.
  Then, with probability $1-\exp(-\Omega_\varepsilon(n))$, we have $S_{\mathcal{C}}< 2$.
\end{lemma}
\begin{proof}
  It suffices to prove this when $n$ is sufficiently large in terms of $\varepsilon$.
  As in \S\ref{sec:review}, let $b_1, b_2, \dots, b_k\in\mathbb{F}_2^n$ be independently and uniformly chosen, and let $\mathcal{C}_i=\spn\{b_1,\dots,b_i\}$.
  Consider the sequence
  \begin{align*}
    \delta_0 \ &\defeq \   2^{-n(1-H(p)-\frac{\eps}{1+\eps})} \nonumber\\
    \delta_i \ &\defeq \   2\delta_{i-1} + \delta_{i-1}^{1.5}.
  \end{align*}
  We can verify by induction that for $i\le n(1-H(p)-\eps)$, we have $\delta_i < 2^{i+1}\delta_0 < 2^{-\frac{\eps^2 n}{2}}$. To see this, notice that the base case is trivial, and if $\delta_j < 2^{j+1}\delta_0$ for $j< i$, we have
  \begin{align}
    \label{eq:rlc-delta}
    \delta_i 
    \ &= \   2\delta_{i-1}(1 + \delta_{i-1}^{0.5}) 
    \ = \   2^i\delta_0\cdot \prod_{j=0}^{i-1}(1 + \delta_{j}^{0.5}) 
    \ \le \ 2^i\delta_0 \cdot \exp\left( \sum_{j=0}^{i-1}\delta_j^{0.5} \right)  
    \ < \  2^{i+1}\delta_0 .
  \end{align}
  In the first two equalities, we applied the definitions of $\delta_i$ and $\delta_{i-1},\dots,\delta_1$, respectively. 
  In the first inequality, we used the estimate $1+z\le e^z$, and in the second we used the inductive hypothesis $\delta_j < 2^{-\frac{\eps^2n}{2}}$ for $j<i$ and that $n$ is sufficiently large.
  By this induction, we conclude that, if $k = n(1-H(p)-\eps)$, then $\delta_k < 2^{-\frac{\eps^2 n}{2}}$.

  Let $b_1,\dots,b_k\in\mathbb{F}_2^n$ be randomly chosen vectors, and let $\mathcal{C}_i=\spn(b_1,\dots,b_i)$ with $\mathcal{C}_k = \mathcal{C}$.
  Call $\mathcal{C}_i$ \emph{good} if $T_{\mathcal{C}_i}< \delta_i$.
  We have
\begin{equation}
  T_{\mathcal{C}_0}
  \ = \ S_{\{0\}}-1 
  \ = \ \left( 2^{\frac{\varepsilon n}{1+\varepsilon}} - 1 \right)\cdot \frac{\Vol(n,pn)}{2^n} 
  \ < \ 2^{\frac{\eps n}{1+\eps}}\cdot\frac{2^{H(p)n}}{2^n} 
  \ = \ \delta_0,
  \label{eq:s0}
\end{equation}
  so $\mathcal{C}_0$ is always good.
  On the other hand, by the definition of $\delta_i$ and Lemma~\ref{lem:ld-t}, if $\mathcal{C}_i$ is good, 
  \begin{align}
    \Pr\left[ \mathcal{C}_{i+1}\text{ not good} \right]
    \ = \ \Pr\left[ T_{\mathcal{C}_{i+1}}\ge \delta_{i+1}\right] 
    \ \le \ \Pr\left[ T_{\mathcal{C}_{i+1}} \ge 2T_{\mathcal{C}_i}+T_{\mathcal{C}_i}^{1.5} \right]
    \ \le \ T_{\mathcal{C}_i}^{0.5} 
    \ < \ \delta_{i}^{0.5}.
  \end{align}
  Thus, with probability at least
  \begin{equation} 
    1 - \left( \delta_0^{0.5} + \delta_1^{0.5} + \cdots + \delta_{k-1}^{0.5} \right) 
    \ > \ 1 - k2^{-\varepsilon^2n/4} 
    \ \ge \ 1 - 2^{-\Omega_\varepsilon(n)}
  \end{equation}
  we have $T_{\mathcal{C}_i} < \delta_i$ for all $i=0,\dots,k$.
  In particular, $T_{\mathcal{C}} = T_{\mathcal{C}_k} < \delta_k < 2^{-\frac{\eps^2 n}{2}}$.
  Thus, $S_\mathcal{C}  = 1 + T_\mathcal{C} < 2$ with probability $1 - \exp(-\Omega_{\eps}(n))$, completing the proof of Lemma~\ref{lem:tech}.
\end{proof}

Finally, we prove the following lemma, which implies Theorem~\ref{thm:ld}.

\begin{lemma}
  Any linear code $\mathcal{C}\subseteq\mathbb{F}_2^n$ of rate $R$ with $S_\mathcal{C}< 2$ is $(p, \frac{1-R}{\varepsilon}+1)$-list decodable.
\label{lem:sc}
\end{lemma}
\begin{proof}
  Suppose for sake of contradiction that there exists $x\sar\in \mathbb{F}_2^n$ such that $|\mathcal{B}(x\sar,p n)\cap \mathcal{C}| > \frac{1-R}{\eps}+1$.
  For all $x\in\mathbb{F}_2^n$ and $c\in \mathcal{C}$, we have 
\[|\mathcal{B}(x+c,p n)\cap \mathcal{C}| = |\mathcal{B}(x,p n)\cap (\mathcal{C}-c)| = |\mathcal{B}(x, p n)\cap \mathcal{C}|,\] 
so $|\mathcal{B}(x\sar+c,p n)\cap \mathcal{C}| > \frac{(1-R)}{\eps}+1$ for all $c \in \mathcal{C}$.
  If $S_\mathcal{C}< 2$, then we have
  \begin{align*}
  2^{n+1} \ > \  2^nS_\mathcal{C} 
  \ &= \ \sum_{x\in \mathbb{F}_2^n} \exp_2\left(n\cdot\frac{\eps}{1+\eps}\cdot |\mathcal{B}(x,p n)\cap \mathcal{C}|\right) \\
  \ &\ge \ \sum_{c\in \mathcal{C}} \exp_2\left(n\cdot\frac{\eps}{1+\eps}\cdot |\mathcal{B}(x\sar+c,p n)\cap \mathcal{C}|\right) \\
  \ &\ge \ \sum_{c\in \mathcal{C}} \exp_2\left(n\cdot\frac{\eps}{1+\eps}\cdot \left(\frac{1-R}{\eps}+1\right)\right) \\
  \ &=   \ |\mathcal{C}|\cdot \exp_2\left(n\cdot\frac{1-R+\eps}{1+\eps}\right) \\
  \ &= \ \exp_2\left(n\left(1 + \frac{\eps R}{1+\eps}\right)\right)  .
  \end{align*}
  which is a contradiction for large enough $n$.
\end{proof}
\begin{proof}[Proof of Theorem~\ref{thm:ld}]
  Apply Lemma~\ref{lem:tech} and then Lemma~\ref{lem:sc} for $R=1-H(p)-\varepsilon$.
\end{proof}

\begin{remark}
  We do not see how to extend this proof to larger alphabets.
  If, for example, $q=3$, Lemma~\ref{lem:ld-t} would need to say $\Pr[S_{\mathcal{C}+\{0,b,2b\}}> 1+3T_\mathcal{C}+o(T_\mathcal{C})] < o(1)$.
  However, the current proof would fail to show this, as we could not separate the expectation in \eqref{eq:ld-b-1}; that is, we cannot say
  \begin{align}
    \E_{b,x}&\left[(A_\mathcal{C}(x)-1) (A_{\mathcal{C}}(x+b)-1)(A_{\mathcal{C}}(x+2b)-1)\right]  \nonumber\\
     &=  \E_{b,x}\left[ A_\mathcal{C}(x)-1 \right] \cdot \E_{b,x}\left[A_\mathcal{C}(x+b)-1\right]\cdot \E_{b,x}\left[A_\mathcal{C}(x+2b)-1\right] 
  \label{}
  \end{align}
\end{remark}

\section{Characterizing the list size distribution}\label{sec:listsize}
In this section we establish Theorem~\ref{thm:listsize}.
For a code $\mathcal{C}$, let 
\begin{align*}
  P_\mathcal{C}\ind{\ge\ell} \ &\defeq \   \E_{x\sim \mathbb{F}_2^n}\left[\mathbb{I}[L_\mathcal{C}(x)\ge\ell]\cdot L_{\mathcal{C}}(x)\right].
\end{align*}
In this way, we have $P_\mathcal{C}\ind{\ge\ell}\ge \Pr[L_\mathcal{C}(x)\ge \ell]$ for $\ell\ge 1$. 
Thus, to prove, Theorem~\ref{thm:listsize}, it suffices to prove the following.
\begin{theorem}
  \label{thm:listsize-2}
  Fix $L\ge 1$ and $\gamma \in (0,1)$.
  For all sufficiently large $n$, if $\mathcal{C} \subseteq \mathbb{F}_2^n$ is a random linear code of dimension $k$, then with probability $1-\exp(-\Omega(\gamma n))$, for all $1\le \ell\le L$,
  \begin{align}
    P_\mathcal{C}\ind{\ge\ell}
    \ \le \ \left(2^{-n(1-H(p))}\cdot 2^k\right)^\ell \cdot 2^{\gamma\ell^2n}.
  \label{eq:gen-func-1}
  \end{align}
\end{theorem}

\begin{proof}[Proof of Theorem~\ref{thm:listsize-2}]
The proof of Theorem~\ref{thm:listsize-2} proceeds by induction on $k$.
For the base case $k=0$, we show \eqref{eq:gen-func-1} holds with probability 1.
For $\ell \geq 2$, \eqref{eq:gen-func-1} holds because then $P_{\{0\}}\ind{\geq \ell} = 0$.
When $\ell=1$, we then have $P_{\{0\}}\ind{\geq 1} = 2^{-n}\cdot\Vol(n,pn) \leq 2^{-n(1 - H(p))}$.

Now that we have established the base case of $k=0$, we proceed by induction.
Lemma~\ref{lem:gen-func-1} provides the inductive step; similar to the approach in \S\ref{sec:ld}, it shows that at every step the ``expected behavior" holds with high probability.
  \begin{lemma}
    \label{lem:gen-func-1}
	Let $\gamma > 0$, and suppose that $\mathcal{C} \subseteq\mathbb{F}_2^n$ is a linear code of dimension $k$ such that for all $1\le \ell\le L$, we have
    \begin{align}
      P_\mathcal{C}\ind{\ge\ell} \ \le \ \left(2^{-n(1-H(p))}\cdot 2^k\right)^\ell \cdot 2^{\gamma\ell^2n}.
      \label{eq:gen-func-lem}
    \end{align}
    Then, for a uniformly chosen $b\in\mathbb{F}_2^n$, with probability at least $1- 2L^3\cdot 2^{-2\gamma n}$ over the choice of $b$, we have, for all $1\le \ell\le L$,
    \begin{align}
      P_{\mathcal{C}+\{0,b\}}\ind{\ge\ell} \ \le \ \left(2^{-n(1-H(p))}\cdot 2^{k+1}\right)^\ell \cdot 2^{\gamma\ell^2n}.
      \label{eq:gen-func-lem-2}
    \end{align}
  \end{lemma}
  \begin{proof}
    For all codes $\mathcal{C}\subset \{0,1\}^n$, we have
    \begin{align}
      P_\mathcal{C}\ind{\ge 1}
      \ &= \ 
      2^{-n}\sum_{i\ge 1}^{} i\cdot |\{x:L_\mathcal{C}(x)=i\}|
      \ = \ 2^{-n}|\{(x,c):x\in\mathbb{F}_2^n, c\in \mathcal{C}, \Delta(x,c)\le pn\}|\nonumber\\
      \ &= \ 2^{-n}\cdot |\mathcal{C}|\cdot \Vol(n,pn)
      \ \le \ 2^{-n(1-H(p))} \cdot 2^k.
    \end{align}
    so \eqref{eq:gen-func-lem}, and thus \eqref{eq:gen-func-lem-2} (which is just \eqref{eq:gen-func-lem} for $C+\{0,b\}$), always holds for $\ell = 1$.

    Now fix a code $\mathcal{C}$ satisfying \eqref{eq:gen-func-lem}, and assume that $\ell\ge 2$.
    We have, for any $b \in \mathbb{F}_2^n$,
    \begin{align}
      P_{\mathcal{C}+\{0,b\}}\ind{\ge\ell}
      \ &= \ \E_{x}\left[ \mathbb{I}[L_{\mathcal{C}+\{0,b\}}(x)\ge \ell]\cdot L_{\mathcal{C}+\{0,b\}}(x) \right] \nonumber \\ 
      \ &\leq \ \E_{x}\left[ \mathbb{I}[L_{\mathcal{C}}(x) + L_{\mathcal{C}}(x+b)\ge \ell] \cdot (L_{\mathcal{C}}(x) + L_{\mathcal{C}}(x+b)) \right] \nonumber \qquad \text{(By Lemma~\ref{lem:AL-bound})} \\ 
      \ &= \ \E_{x}\left[ \mathbb{I}[L_{\mathcal{C}}(x) + L_{\mathcal{C}}(x+b)\ge \ell] \cdot L_{\mathcal{C}}(x)\right] +  \E_x\left[\mathbb{I}[L_{\mathcal{C}}(x) + L_{\mathcal{C}}(x+b)\ge \ell]\cdot L_{\mathcal{C}}(x+b)\right] \nonumber\\ 
      \ &= \ 2\cdot \E_{x}\left[ \mathbb{I}[L_{\mathcal{C}}(x) + L_{\mathcal{C}}(x+b)\ge \ell] \cdot L_{\mathcal{C}}(x) \right] \qquad \left(\text{Set $x'=x+b$ in second $\E$}\right)\nonumber\\
      \ &= \ 2\cdot \E_{x}\left[ \mathbb{I}[L_{\mathcal{C}}(x) \ge \ell] \cdot L_{\mathcal{C}}(x)  + \sum_{i=0}^{\ell-1} \mathbb{I}[L_\mathcal{C}(x)=i, L_\mathcal{C}(x+b)\ge\ell-i]\cdot L_\mathcal{C}(x) \right] \nonumber\\
      \ &= \ 2\cdot P_\mathcal{C}\ind{\ge\ell} + 2\cdot \sum_{i=0}^{\ell-1} i\cdot \E_{x}\left[  \mathbb{I}[L_\mathcal{C}(x)=i]\cdot \mathbb{I}[L_\mathcal{C}(x+b)\ge\ell-i]\right]  \nonumber\\
      &=: (\star).
      \label{eq:gen-func-star}
    \end{align}
Taking expectations of $(\star)$ over a random $b$, we have,
    \begin{align*}
      \E_{b\sim\mathbb{F}_2^n}\left[ (\star) \right]
      \ &= \ 2\cdot P_\mathcal{C}\ind{\ge\ell} + 2\cdot \sum_{i=0}^{\ell-1} i\cdot \E_{b,x}\left[  \mathbb{I}[L_\mathcal{C}(x)=i]\cdot \mathbb{I}[L_\mathcal{C}(x+b)\ge\ell-i]\right] \nonumber\\
      \ &= \ 2\cdot P_\mathcal{C}\ind{\ge\ell} + 2\cdot \sum_{i=0}^{\ell-1} i\cdot \Pr_{b,x}  [L_\mathcal{C}(x)=i]\cdot\Pr_{b,x}[L_\mathcal{C}(x+b)\ge\ell-i] \nonumber\\
      \ &< \ 2\cdot P_\mathcal{C}\ind{\ge\ell} + 2L\cdot \sum_{i=0}^{\ell-1} P_\mathcal{C}\ind{\ge i}\cdot P_\mathcal{C}\ind{\ge\ell-i}
    \end{align*}
    The second equality uses that $x$ and $b$ are independent and uniform on $\mathbb{F}_2^n$ so $x$ and $x+b$ are independent.
    The inequality uses that $i < L$, $\Pr_{b,x}[L_\mathcal{C}(x)=i]\le P_\mathcal{C}\ind{\ge i}$, and $\Pr_{b,x}[L_{\mathcal{C}}(x+b)=\ell-i]\le P_\mathcal{C}\ind{\ge \ell-i}$.
    Continuing, we bound
    \begin{align*}
      \E_b\left[ (\star) - 2\cdot P_\mathcal{C}\ind{\ge\ell} \right]
      \ &< \   2 L \cdot \sum_{i=1}^{\ell-1} P_\mathcal{C}\ind{\ge i}\cdot P_\mathcal{C}\ind{\ge\ell-i} \nonumber\\
      \ &\le \  2 L \cdot \sum_{i=1}^{\ell-1} \left(2^{-n(1-H(p))}\cdot 2^k\right)^{i+(\ell-i)}\cdot2^{\gamma (i^2+(\ell-i)^2)n} \nonumber & \text{(By \eqref{eq:gen-func-lem})}\\
      \ &\le \  2 L^2\cdot \left(2^{-n(1-H(p))}\cdot 2^{k}\right)^\ell\cdot2^{\gamma(\ell^2-2)n}
    \end{align*}
    By definition of $(\star)$, we have $(\star)-2P_\mathcal{C}\ind{\ge \ell}\ge 0$ always, so we may apply Markov's inequality to obtain
      \[\Pr_b\left[ (\star) -2P\ind{\ge\ell}_\mathcal{C} \ge \left(2^{-n(1-H(p))}\cdot 2^{k}\right)^\ell\cdot2^{\gamma\ell^2n} \right] \ \le \   2L^2\cdot 2^{-2\gamma n}. \]
      As $P_{\mathcal{C}+\{0,b\}}\ind{\ge \ell}\le (\star)$ by \eqref{eq:gen-func-star}, we have
      \begin{equation}\Pr_b\left[ P_{\mathcal{C}+\{0,b\}}\ind{\ge\ell} -2P\ind{\ge\ell}_\mathcal{C} \ge \left(2^{-n(1-H(p))}\cdot 2^{k}\right)^\ell\cdot2^{\gamma\ell^2n} \right] \ \le \   2L^2\cdot 2^{-2\gamma n}.
    \label{eq:gen-func-3}
    \end{equation}
    Thus, with probability at least $1-2L^22^{-2\gamma n}$, we have
    \begin{align}
      P_{\mathcal{C}+\{0,b\}}\ind{\ge\ell} 
      \ &\le \   2P\ind{\ge\ell}_\mathcal{C} + \left(2^{-n(1-H(p))}\cdot 2^{k}\right)^\ell\cdot2^{\gamma\ell^2n} \nonumber\\
      \ &\le \   3\left(2^{-n(1-H(p))}\cdot 2^{k}\right)^\ell\cdot2^{\gamma\ell^2n}   \nonumber\\
      \ &\le \   \left(2^{-n(1-H(p))}\cdot 2^{k+1}\right)^\ell\cdot2^{\gamma\ell^2n}.
      \label{eq:gen-func-4}
    \end{align}
    The second inequality uses the assumption \eqref{eq:gen-func-lem} and the final inequality uses $\ell\ge 2$.
    Union bounding over all $1\le \ell\le L$, we have \eqref{eq:gen-func-4} holds for all $1\le\ell\le L$ with probability $1-2L^32^{-2\gamma n}$.
    This completes the proof of Lemma~\ref{lem:gen-func-1}.
  \end{proof}

  Returning to the proof of Theorem~\ref{thm:listsize-2}, call a code $\mathcal{C}$ of dimension $k$ \emph{good} if \eqref{eq:gen-func-1} holds for all $1\le \ell\le L$.
  Lemma~\ref{lem:gen-func-1} states that if $\mathcal{C}$ is good, then $\mathcal{C} + \{0,b\}$ fails to be good with probability at most $2L^32^{-2\gamma n}$ over the choice of a uniformly random $b \in \mathbb{F}_2^n$.

Since we have already shown that $\{0\}$ is good with probability 1 at the beginning of the proof, it follows from the union bound that a random linear binary code $\mathcal{C} = \spn(b_1,\dots,b_k)$ fails to be good with probability at most $k\cdot 2L^32^{-2\gamma n}$, which is less than $2^{-\gamma n}$ for $n$ sufficiently large.  This completes the proof of Theorem~\ref{thm:listsize-2}.
\end{proof}

To finish the section, we demonstrate how Theorem~\ref{thm:listsize} implies Theorem~\ref{thm:ld}.
\begin{corollary}[Theorem~\ref{thm:ld}]
  \label{cor:listsize-ld}
\thmld
\end{corollary}
\begin{proof}
  Let $L = H(p)/\eps + 2$ and choose $\gamma = \frac{\varepsilon}{2L^2}$, and let $\mathcal{C} \subseteq \mathbb{F}_2^n$ be a random linear code of dimension $k = n( 1 - H(p) - \eps)$.
By Theorem~\ref{thm:listsize}, with probability at least $1-\exp(-\gamma n)$, code $\mathcal{C}$ satisfies \eqref{eq:gen-func-1} for $k = n(1 - H(p) - \eps)$ and all $1\le\ell\le L$.
  Choosing $\ell = L$, we have 
  \begin{align*}
  \Pr_{x\sim \mathbb{F}_2^n}[L_\mathcal{C}(x)\ge L]
  \ \le \ 2^{L(-n(1-H(p))+k) + \gamma L^2n}
  \ = \ 2^{-nL\eps + \frac{\varepsilon}{2} n}
  \ < \ 2^{-(n-k)}
  \end{align*}
  where the last equality uses the definition of $L$ and $\gamma$.
  Suppose there exists $x \in \mathbb{F}_2^n$ such that $L_\mathcal{C}(x)\ge L$. Then $L_\mathcal{C}(x+c)\ge L$ for all $c\in C$, so that $\Pr[L_\mathcal{C}(x)\ge L]\ge 2^{-(n-k)}$.
  This is a contradiction, so there are no $x \in \mathbb{F}_2^n$ such that $L_{\mathcal{C}}(x)\ge L$.
\end{proof}

\section{List-decoding rank metric codes}\label{sec:rank-metric}
Here, we prove Theorem~\ref{thm:rank-metric}, which is restated below.
\begin{theorem}[Theorem~\ref{thm:rank-metric}, restated]
  \thmrankmetric
\end{theorem}

  Let $k = Rmn$.
  Let $Y_1,\dots, Y_k\in\mathbb{F}_2^{m\times n}$ be i.i.d. uniformly random matrices, and let $\mathcal{C}_i = \spn(Y_1,\dots,Y_i)$ so that $\mathcal{C}_k$ is a random linear rank metric code of rate $R$.
  We overload the notation from \S\ref{sec:ld} to be matrix-valued:
  for a rank metric code $\mathcal{C}\subseteq\mathbb{F}_2^{m\times n}$ and $X\in\mathbb{F}_2^{m\times n}$, let
  \begin{align*}
    L_\mathcal{C}(X) \ &\defeq \ |\mathcal{B}_R(X,p)\cap \mathcal{C}|  \nonumber\\
    A_\mathcal{C}(X) \ &\defeq \  \exp_2\left(\frac{\eps}{1+\eps} mn L_\mathcal{C}(X)\right) \nonumber\\
    S_\mathcal{C} \ &\defeq \    \E_{X\sim\mathbb{F}_2^{m\times n}} \left[ A_\mathcal{C}(X) \right] \nonumber\\
    T_\mathcal{C} \ &\defeq \ S_\mathcal{C} - 1.  
  \label{}
  \end{align*}
Using this notation, we have the following analog to Lemma~\ref{lem:AL-bound}, the proof of which follows identically to that of Lemma~\ref{lem:AL-bound}.
    \begin{lemma}
      \label{lem:rank-AL-bound}
      We have, for all linear $\mathcal{C}$ and all $Y\in\mathbb{F}_2^{m\times n}$,
      \begin{align}
        L_{\mathcal{C}+\{0,Y \}}(X) \ &\le \   L_\mathcal{C}(X) + L_\mathcal{C}(X+Y) \label{eq:rank-q-bound-1} \\
        A_{\mathcal{C}+\{0,Y\}}(X) \ &\le \   A_\mathcal{C}(X) \cdot A_\mathcal{C}(X+Y),
        \label{eq:rank-q-bound-2}
      \end{align}
      with equality if and only if $Y\notin\mathcal{C}$.
    \end{lemma}

  We need the following approximation of the size of a rank metric ball.
  We state it for general $q$ for an application in Appendix~\ref{app:rank-metric-lb}, although we only use the binary case here.
  \begin{lemma}[\cite{GadouleauY08}]
    \label{lem:rank-metric-vol}
    Let $X\in\mathbb{F}_q^{m\times n}$. Then
    \begin{align}
      q^{mn(p+pb-p^2b)} \ \le \ |\mathcal{B}_R(X,p)| \ < \ 4\cdot q^{mn(p+pb-p^2b)}
    \label{}
    \end{align}
  \end{lemma}
Lemma~\ref{lem:rank-metric-vol} yields the inequality
  \begin{align}
    S_{\{0\}} \ &\le \ 1 + 4\cdot 2^{-mn(1-p-pb+p^2b - \frac{\varepsilon}{1+\varepsilon})}.
  \label{eq:rank-metric-s0}
  \end{align}

  Following the proof of Theorem~\ref{thm:ld}, we prove the following lemma which mirrors Lemma~\ref{lem:ld-t}.
    \begin{lemma}
      \label{lem:rank-metric-t}
      Suppose that $\mathcal{C}$ is fixed and satisfies $T_\mathcal{C} < 1$, so that $S_{\mathcal{C}}<2$. Then
      \begin{align*}
        \Pr_{Y\sim\mathbb{F}_2^{m\times n} } \left[ S_{\mathcal{C} + \{0,Y\}} \ge 1 + 2T_\mathcal{C} + T_\mathcal{C}^{1.5} \right] \ &\le \   T_\mathcal{C}^{0.5}.
      \label{}
      \end{align*}
    \end{lemma}
    \begin{proof}
      By Lemma~\ref{lem:rank-AL-bound}, for all $Y$,
      \begin{align*}
        S_{\mathcal{C} + \{0,Y\}} 
	\ &= \  \E_X\left[ A_{\mathcal{C} + \{0,Y\}}(X) \right] \\
        \ &\le \   \E_X\left[A_\mathcal{C}(X) A_\mathcal{C}(X+Y)\right] \nonumber\\
        \ &= \   \E_X\left[ -1+ A_\mathcal{C}(X) + A_\mathcal{C}(X+Y) + (A_\mathcal{C}(X)-1)(A_\mathcal{C}(X+Y)-1)\right] \nonumber\\
        \ &= \   1+2T_{\mathcal{C}} + \E_X\left[ (A_\mathcal{C}(X)-1)(A_\mathcal{C}(X+Y)-1)\right].
      \end{align*}
      Over the randomness of $Y$ and $X$, we have $X$ and $X+Y$ are statistically independent and uniform over $\mathbb{F}_2^{m \times n}$, so we have
      \begin{align*}
        \E_Y\E_X\left[(A_\mathcal{C}(X)-1)(A_{\mathcal{C}}(X+Y)-1)\right] 
        \ = \ \E_{Y,X}\left[ A_\mathcal{C}(X)-1 \right] \cdot \E_{Y,X}\left[A_\mathcal{C}(X+Y)-1\right]
        \ = \ T_{\mathcal{C}}^2.
      \end{align*}
      By Markov's inequality,
      \begin{align*}
        \Pr_Y \left[ S_{\mathcal{C} + \{0,Y\}} \ge 1+2T_{\mathcal{C}}+T_{\mathcal{C}}^{1.5} \right]
        \ &\le \   \Pr_Y\left[\E_X\left[(A_\mathcal{C}(X)-1)(A_{\mathcal{C}}(X + Y)-1)\right] \ge T_{\mathcal{C}}^{1.5}\right] 
        \ \le \ \frac{T_{\mathcal{C}}^2}{T_\mathcal{C}^{1.5}} 
        \ = \ T_\mathcal{C}^{0.5}.
      \end{align*}
This proves the lemma.
    \end{proof}

Continuing with the proof structure of Theorem~\ref{thm:ld}, we prove the following lemma, which mirrors Lemma~\ref{lem:tech}.
  \begin{lemma}
    \label{lem:rank-metric-tech}
    Let $p\in(0,1)$, $\varepsilon\in(0,\half)$.
    Let $n < m$ be positive integers with $b=n/m$ and with $n$ sufficiently large.
    With probability $1-\exp(-\Omega_\varepsilon(n))$, a random linear rank metric code $\mathcal{C} \subseteq \mathbb{F}_2^{m \times n}$ of rate $R = (1-p)(1-bp)-\eps$ satisfies $S_\mathcal{C} < 2$.
  \end{lemma}

  \begin{proof}[Proof of Lemma~\ref{lem:rank-metric-tech}]

    As in the proof of Lemma~\ref{lem:tech}, consider the sequence
    \begin{align*}
      \delta_0 \ &\defeq \   4\cdot \exp_2\left(-mn\left((1-p)(1-bp)-\frac{\eps}{1+\eps}\right)\right) \nonumber\\
      \delta_i \ &\defeq \   2\delta_{i-1} + \delta_{i-1}^{1.5}
    \end{align*}
    Like in Lemma~\ref{lem:tech}, we verify by induction that for $i\le mn((1-p)(1-bp)-\eps)$, we have $\delta_i < 2^{i+1}\delta_0 < 2^{-\frac{\eps^2 mn}{2}}$.  The base case is trivial (assuming $mn$ is sufficiently large), and if $\delta_j < 2^{j+1}\delta_0$ for $j< i$, we have
      \begin{align}
        \delta_i 
        \ &= \   2\delta_{i-1}(1 + \delta_{i-1}^{0.5}) 
        \ = \   2^i\delta_0\cdot \prod_{j=0}^{i-1}(1 + \delta_{j}^{0.5}) 
        \ \le \ 2^i\delta_0 \cdot \exp\left( \sum_{j=0}^{i-1}\delta_j^{0.5} \right)  
        \ < \  2^{i+1}\delta_0 .
      \label{}
      \end{align}
      In the first inequality, we used the estimate $1+z\le e^z$, and in the second we used the inductive hypothesis $\delta_j < 2^{-\frac{\eps^2 mn}{2}}$ for $j<i$ and that $mn$ is sufficiently large.
      In particular, if $k = mn((1-p)(1-bp)-\eps)$, then $\delta_k < 2^{-\frac{\eps^2 mn}{2}}$.

      Let $Y_1,\dots,Y_k\in\mathbb{F}_2^{m\times n}$ be randomly chosen matrices, and let $\mathcal{C}_i=\spn(Y_1,\dots,Y_i)$ with $\mathcal{C}_k = \mathcal{C}$.
      By Lemma~\ref{lem:rank-metric-t}, conditioned on a fixed $\mathcal{C}_i$ satisfying $T_{\mathcal{C}_i}< \delta_i$, then with probability at most $T_{\mathcal{C}_i}^{0.5}$, which is at most $\delta_i^{0.5}$, we have $T_{\mathcal{C}_{i+1}} > \delta_{i+1}$.
      Furthermore, $T_{\mathcal{C}_0} < \delta_0$ by \eqref{eq:rank-metric-s0}.
      Thus, with probability at least
      \begin{align}
        1 - \left( \delta_0^{0.5} + \delta_1^{0.5} + \cdots + \delta_{k-1}^{0.5} \right) 
        \ > \ 1 - ke^{-\varepsilon^2mn/4} 
        \ \ge \ 1 - 2^{-\Omega(mn)}
      \label{}
      \end{align}
      we have $T_{\mathcal{C}_i} < \delta_i$ for all $i$.
      In particular, this implies $T_{\mathcal{C}} = T_{\mathcal{C}_k} < \delta_k < 2^{-\frac{\eps^2 mn}{2}}$.
      Thus, we have $S_\mathcal{C}  = 1 + T_\mathcal{C} < 2$ with probability $1 - \exp(-\Omega_{\eps}(n))$, completing the proof of Lemma~\ref{lem:rank-metric-tech}.
  \end{proof}
  
We now can finish the proof of Theorem~\ref{thm:rank-metric}.
  \begin{proof}[Proof of Theorem~\ref{thm:rank-metric}]
  By Lemma~\ref{lem:rank-metric-tech}, it suffices to prove that $S_\mathcal{C} < 2$ implies $(p,\frac{p+bp-bp^2}{\eps}+2)$-list decodability.
  Suppose for sake of contradiction that a code $\mathcal{C}$ satisfies $S_\mathcal{C}< 2$ and there exists $X\sar\in \mathbb{F}_2^{m \times n}$ such that $|\mathcal{B}_R(X\sar,p)\cap \mathcal{C}| > \frac{p+bp-bp^2}{\eps}+2$.
  We know that for all $C\in \mathcal{C}$, we have $|\mathcal{B}_R(X^*+C,p)\cap \mathcal{C}| = |\mathcal{B}_R(X^*,p)\cap (\mathcal{C}-C)| = |\mathcal{B}_R(X^*, p)\cap \mathcal{C}|$, so $|\mathcal{B}_R(X\sar+C,p)\cap \mathcal{C}| > \frac{p+bp-bp^2}{\eps}+2$ for all $C\in \mathcal{C}$.

  If $S_\mathcal{C}< 2$, then we have
  \begin{align}
  2^{mn+1} \ > \ 2^{mn}\cdot S_\mathcal{C} 
  \ &= \ \sum_{X\in \mathbb{F}_2^{m \times n}} \exp_2\left(mn\cdot\frac{\eps}{1+\eps}\cdot |\mathcal{B}_R(X,p)\cap \mathcal{C}|\right) \nonumber\\
  \ &\ge \ \sum_{C\in \mathcal{C}} \exp_2\left(mn\cdot\frac{\eps}{1+\eps}\cdot |\mathcal{B}_R(X\sar+C,p)\cap \mathcal{C}|\right) \nonumber\\
  \ &\ge \ \sum_{C\in \mathcal{C}} \exp_2\left(mn\cdot\frac{\eps}{1+\eps}\cdot \left(\frac{p+bp-bp^2}{\eps}+2\right)\right) \nonumber\\
  \ &=   \ |\mathcal{C}|\cdot \exp_2\left(mn\cdot\frac{p+bp-bp^2+2\eps}{1+\eps}\right) \nonumber\\
  \ &= \ \exp_2\left(mn\left(1 + \frac{\eps}{1+\eps}((1-p)(1-bp)-\eps)\right)\right)  .
  \end{align}
  which is a contradiction for large enough $mn$.
  \end{proof}

\section{Conclusion}
In this work, we have given an improved analysis of the list-decodability of random linear binary codes.  Our analysis 
works for all values of $p$, and also obtains improved bounds on the list size as the rate approaches list-decoding capacity.  In particular, not only do our bounds improve on previous work for random linear codes, but they show that random linear codes are more list-decodable than completely random codes, in the sense that the list size is strictly smaller. 
Our techniques are quite simple, and strengthen an argument of \cite{GuruswamiHSZ02} to hold with high probability.  In order to demonstrate the applicability of these techniques, we use them to (a) obtain more information about the distribution of list sizes of random linear codes and (b) to prove a similar result for random linear rank-metric codes, improving a recent result of \cite{GuruswamiR17}.

We end with some open questions raised by our work.
\begin{enumerate}
\item With the exception of Theorem~\ref{thm:rank-metric-lb}, our results---both our upper bounds and our lower bounds---hold only for binary alphabets.  We conjecture that analogous results, and in particular list decoding random linear codes with list size $C/\varepsilon$ for $C<1$, hold over larger alphabets. 
\item We showed that random linear binary codes of rate $1 - H(p) - \eps$ are with high probability $(p,L)$ list-decodable with $L \leq H(p)/\eps$.  The lower bounds of \cite{Rudra11, GuruswamiN14}  show that we must have $L \geq C/\eps$, but the constant $C$ is much smaller than $H(p)$.  Thus, we still do not know what the correct leading constant is for random linear codes.
\item Finally, there are currently no known explicit constructions of capacity-achieving binary list-decodable codes for general $p$.  It is our hope that this work---which gives more information about the structure of linear codes which achieve list-decoding capacity---could lead to progress on this front.  Given that we don't know how to efficiently check if a given code is $(p,L)$-list-decodable, even an efficient Las Vegas construction (as opposed to a Monte Carlo construction) would be interesting.
\end{enumerate}

\section*{Acknowledgements}
We thank anonymous reviewers for helpful comments on a previous version of this manuscript.

\bibliographystyle{alpha}
\newcommand{\etalchar}[1]{$^{#1}$}

\appendix

\section{Lower bound on the list size of random codes}\label{app:lb}
In this appendix, we prove 
Theorem~\ref{thm:lb}, which states
that, for any $\gamma > 0$ and $\eps$ small enough, the smallest list size $L$ such that a completely random code is with high probability $(p,L)$-list decodable, satisfies $L\in [(1-\gamma)/\eps-1, 1/\eps]$.
In other words, the list size of $1/\eps$ given by the classical list decoding capacity theorem is tight in the leading constant factor.
We prove this by tightening the second moment argument of Guruswami and Narayanan \cite{GuruswamiN14}.
We follow the exact same proof outline, differing only in how we bound one expression.  This improvement appears as Lemma~\ref{lem:lb-int}.
In both Appendices~\ref{app:lb} and \ref{app:rank-metric-lb}, we use the following useful fact.
\begin{fact}\label{fact:secondmoment}
  Let $Z$ be a random variable that takes only nonnegative values. Then
  \begin{align*}
    \Pr[Z = 0] \ \le \ \frac{\Var Z}{\E[Z]^2}.
  \end{align*}
\end{fact}

For this appendix, we assume that $pn$ is an integer.
We need a few bounds on Hamming balls and their intersections.  We first use the following classical estimate on the volume of the Hamming ball (see, for example, \cite{MacwilliamsS77}[pp. 308-310]).
\begin{lemma}
  \label{lem:vol-bound-1}
  Suppose that $1\le pn\le n/2$. Then
  \begin{align}
    \frac{2^{H(p)n}}{\sqrt{8np(1-p)}} 
    \ &\le \ \binom{n}{pn}
    \ \le \  \Vol(n,pn) 
    \ \le \ 2^{H(p)n}.
  \label{}
  \end{align}
\end{lemma}
We also need a bound on the volume of the \em intersections \em of two Hamming balls whose centers are distance $d$ apart, which we denote by
\[    
\Vol(n,pn;d) \ \defeq \ |\mathcal{B}(0,pn)\cap \mathcal{B}(1^d0^{n-d},pn)|.
\]
Lemma~\ref{lem:vol-bound-2} below bounds $\Vol(n,pn;d)$ above.
\begin{lemma}
  \label{lem:vol-bound-2}
  Suppose $0 < p < 1/2$.
  Let
  \begin{align*}
    \alpha_{p} \ \defeq \ \half \log_{2}\frac{1}{4p(1-p)}.
  \label{}
  \end{align*}
  There exists a constant $C_p > 0$ depending only on $p$ such that, for sufficiently large $n$, for all $1\le d \le n$, we have
  \begin{align}
    \frac{\Vol(n,pn;d)}{\Vol(n,pn)}
    \ &\le \  2^{-\alpha_{p}d} \cdot C_p\sqrt{n}.
\label{eq:intersection}
  \end{align}
\end{lemma}
\begin{proof}
  Let $w=1^d0^{n-d}$ be a string of weight $d$, so that we  
  wish to upper bound the number of points $x$ in $D \defeq \mathcal{B}(0,pn)\cap \mathcal{B}(w,pn)$.

First, notice that if $pn < d/2$, then $\mathcal{B}(0,pn) \cap \mathcal{B}(w, pn) = \emptyset$, so \eqref{eq:intersection} holds.  Thus, assume that $pn \geq d/2$.
  Suppose that $x$ restricted to the first $d$ coordinates has weight $i$, and that $x$ restricted to the last $n-d$ coordinates has weight $j$.
Such an $x$ is in $D$ if and only if 
$(d-i)+j\le pn$ and $i+j\le pn$.
For a fixed $i,j$, the number of such $x$ is 
exactly $\binom{d}{i}\binom{n-d}{j}$.
 Thus,
  \begin{align*}
    \Vol(n,pn;d)
    \ &= \ \sum_{i=0}^d\sum_{j=0}^{pn-\max(i,d-i)} \binom{d}{i}\binom{n-d}{j} \nonumber\\
    \ &\le \ \sum_{i=0}^d\binom{d}{i}\sum_{j=0}^{pn-d/2} \binom{n-d}{j} \\
    \ &= \ \sum_{i=0}^d\binom{d}{i}\Vol(n-d, np-(d/2)) \nonumber\\
    \ &= \ 2^d\cdot \Vol(n-d, np-(d/2)) \nonumber\\
    \ &\le \ \exp_2\left(d + H_2\left( \frac{pn-\frac{d}{2}}{n-d} \right)\cdot (n-d) \right),
  \label{}
  \end{align*}
using Lemma~\ref{lem:vol-bound-1} in the final line.
  Define
  \begin{align*}
    f(d) \ &\defeq \ d + H_2\left( \frac{pn-\frac{d}{2}}{n-d} \right)\cdot (n-d).
  \label{}
  \end{align*}
When $d = 0$, \eqref{eq:intersection} holds for sufficiently large $n$.  Thus, we show that the derivative $f'(d)$ is sufficiently negative that \eqref{eq:intersection} continues to hold as $d$ increases.  
\begin{claim} \label{claim:derivative} We have $f'(d) \le -\alpha_{p}$.  \end{claim}
\begin{proof}
Let $g(x) = x\log_2(x)$. We know that $g'(x) = \log_2(x) + \log_2(e)$.
We compute that
  \begin{align}
    \label{eq:deriv-1}
    f(d) \ &= \ d + (pn-d/2)\log_2\left(\frac{n-d}{pn-d/2} \right)+ (n-pn-d/2)\log_2\left(\frac{n-d}{n-pn-d/2}\right)\nonumber\\ 
    \ &= \ d + g(pn-d/2) + g(n-pn-d/2) - g(n-d)\nonumber\\ 
    f'(d) \ &= \ 1 + \half\log_2\frac{(pn-d/2)(n-pn-d/2)}{(n-d)^2}.
  \end{align}
  We also know that
  \begin{align}
    \label{eq:deriv-2}
    \frac{(pn-d/2)(n-pn-d/2)}{(n-d)^2} \ \le \  p(1-p).
  \end{align}
Indeed, cross multiplying and expanding \eqref{eq:deriv-2} yields $p(1-p)n^2 - \frac{nd}{2} + \frac{d^2}{4} \leq (n^2 - 2nd + d^2)p(1 - p)$, which simplifies to $p(1 - p) \leq \frac{ n/2 - d/4 }{2n - d} =1/4$, which is always true.  
%%%For completeness, here's a bit more detail --mary :
%%%$\frac{(pn - d/2)( (1-p)n - d/2 )}{(n - d)^2} \leq p(1-p)$ is equivalent to
%%%\begin{align*}
%%%p(1-p)n^2 - \frac{nd}{2} + \frac{d^2}{4} &\leq (n^2 - 2nd + d^2)p(1 - p) \\
%%%p(1 - p) &\leq \frac{ n/2 - d/4 }{2n - d} = \frac{1}{4},
%%%\end{align*}
%%%which is true.
Since the left-hand side is always positive (using the assumption $p < 1/2$ and $d/2 \leq pn$), we may take logs of both sides and conclude that
  \begin{align*}
    f'(d) 
    \ &\le \ 1 + \half\log_2 p(1-p) 
    \ = \ -\alpha_{p} .\qedhere
  \label{}
  \end{align*}
\end{proof}
  Returning to the proof of Lemma~\ref{lem:vol-bound-2}, as $f(0) = H(p)n$ and $f'(d)\le -\alpha_p$ by Claim~\ref{claim:derivative}, we have $f(d) \le H(p)n-\alpha_{p}d$ for all $d\ge 0$, and in particular all integers $d\ge0$.
  Taking $C_p = \sqrt{8p(1-p)}$, and applying Lemma~\ref{lem:vol-bound-1}, we conclude
  \begin{align*}
    \frac{\Vol(n,pn;d)}{\Vol(n,pn)}
    \ &\le \ \frac{\exp_2(H(p)n -\alpha_p d)}{\exp_2(H(p)n) / \sqrt{8np(1-p)}} 
    \ \le \  2^{-\alpha_{p}d} \cdot C_p\sqrt{n}.\qedhere
  \end{align*}
\end{proof}

In Lemma~\ref{lem:lb-int}, we record computation used in the proof of Theorem~\ref{thm:lb}.
The proof of Theorem~\ref{thm:lb} uses the second moment method to show that, with high probability there are a positive number of ``bad events'', by showing the expectation of the number of bad events (denoted $W$) is large, and the variance $\Var[W]$ of the number of bad events is comparatively small.
The bound in Lemma~\ref{lem:lb-int} gives a tighter bound on individual terms in the expansion of $\Var[W]$ than the corresponding bound used in \cite{GuruswamiN14}.
This helps us obtain a better overall list size lower bound.
\begin{lemma}
  \label{lem:lb-int}
  Let $0<p<1/2$, $1\le \ell \le L$ be integers and $\mu \defeq 2^{-n}\Vol(n,pn)$.
  Let $C_p$ and $\alpha_p$ be given by Lemma~\ref{lem:vol-bound-2}.
  Let 
  \begin{align}
    a,b, x_1,\dots, x_\ell, x_{\ell+1},\dots, x_L, y_{\ell+1},\dots,y_L 
  \end{align}
  be chosen independently and uniformly at random from $\mathbb{F}_2^n$ (so there are $2 + 2L-\ell$ points total).
  Let $\calE$ be the event
  \begin{align*}
    \calE \ &= \    \left(\bigwedge_{i=1}^\ell [\Delta(a,x_i) \le pn] \wedge [\Delta(b,x_i)\le pn] \right)\bigwedge \left( \bigwedge_{i=\ell+1}^L [\Delta(a,x_i) \le pn] \wedge [\Delta(b,y_i)\le pn] \right).
  \end{align*}
  Then
  \begin{align}
    \Pr[\calE] 
    \ \le \ \min\left(\mu^{2L-\ell+1}, 2^{-n}\cdot \mu^{2L-\ell}\cdot C_p^L\cdot n^{L/2}\cdot (1 + 2^{-\alpha_{p}\ell})^n\right).
  \label{eq:lb-int}
  \end{align}
\end{lemma}
\begin{proof}
  First we show the probability of $\calE$ is at most $\mu^{2L-\ell+1}$. This is the argument that was presented in \cite{GuruswamiN14}.
  For the event $\calE$ to occur, we need the following necessary (but not necessarily sufficient) conditions: (1) $\Delta(a,x_1)\le pn$ and $\Delta(b,x_1)\le pn$, (2) $\Delta(a,x_i)\le pn$ for $i=2,\dots,L$, and (3) $\Delta(b,y_i)\le pn$ for $i=\ell+1,\dots,L$.  Notice that each of (1),(2),(3) are independent.
  Using the randomness of $a$ and $b$, (1) happens with probability $\mu^2$.  Conditioned on the positions of $a$ and $b$, (2) and (3) occur with probability $\mu^{L-1}$ and $\mu^{L-\ell}$, respectively.
  Thus, we have $\Pr\left[\mathcal{E}\right] \leq \mu^{2L-\ell+1}$.

  For the other bound, we first condition on the locations of $a$ and $b$. Note that, for $0\le d\le n$, $a$ and $b$ have Hamming distance exactly $d$ with probability $2^{-n}\binom{n}{d}$. 
  For $i=1,\dots,\ell$, conditioned on the positions of $a$ and $b$ being Hamming distance $d$ apart, the probability that $\Delta(a,x_i)\le pn$ and $\Delta(b,x_i)\le pn$ is exactly $2^{-n}\Vol(n,pn;d)$. 
  For $i=\ell+1,\dots,L$, the probability that $\Delta(a,x_i)\le pn$ and the probability that $\Delta(b,y_i)\le pn$ are each exactly $\mu$.
  We thus can write, using Lemma~\ref{lem:vol-bound-2},
  \begin{align}
    \Pr[\calE]
    \ &= \ \sum_{d=0}^n \Pr[\Delta(a,b)=d]\cdot \left( \frac{\Vol(n,pn;d)}{2^n} \right)^\ell \cdot \mu^{2L-2\ell} \nonumber\\
    \ &= \ \sum_{d=0}^n 2^{-n}\binom{n}{d} \cdot \left( \frac{\Vol(n,pn;d)}{\Vol(n,pn)}\cdot\frac{\Vol(n,pn)}{2^n} \right)^\ell \cdot \mu^{2L-2\ell} \nonumber\\
    \ &= \ \sum_{d=0}^n 2^{-n}\binom{n}{d} \cdot \left( \frac{\Vol(n,pn;d)}{\Vol(n,pn)}\right)^\ell \cdot \mu^{2L-\ell} \nonumber\\
    \ &= \ 2^{-n}\mu^{2L-\ell}\sum_{d=0}^n \binom{n}{d} \cdot \left( \frac{\Vol(n,pn;d)}{\Vol(n,pn)}\right)^\ell \nonumber\\
    \ &\le \ 2^{-n}\mu^{2L-\ell}\sum_{d=0}^n \binom{n}{d} \cdot \left( 2^{-\alpha_{p}d}\cdot C_p\sqrt{n}\right)^\ell \nonumber\\
    \ &\le \ 2^{-n}\cdot \mu^{2L-\ell}\cdot C_p^L\cdot n^{L/2}\cdot (1 + 2^{-\alpha_{p}\ell})^n,
  \end{align}
  as desired.
\end{proof}

We now prove our theorem.  
\begin{theorem}[Theorem~\ref{thm:lb}, restated]
  \thmlb
\end{theorem}
\begin{proof}
  We follow the outline \cite{GuruswamiN14}.
  Let $\alpha_p$ be as in Lemma~\ref{lem:vol-bound-2}.
  With hindsight, let 
  \begin{align}
    \ell_{p,\varepsilon} \ \defeq \ \frac{1-H(p)}{2\varepsilon},
    \qquad
    \gamma_{p,\varepsilon} \ \defeq \ \frac{1}{2\ln 2}\cdot 2^{-\alpha_p\ell_{p,\varepsilon}} = \exp\left( -\Omega_p\left( \frac{1}{\varepsilon} \right) \right),
  \label{}
  \end{align}
  Let $\mathcal{C} \subseteq \mathbb{F}_2^n$ be a uniformly random code of rate $R=1-H(p)-\eps$.  We think of each codeword $c \in \mathcal{C}$ as the encoding of a distinct message $x \in \mathbb{F}_2^{Rn}$.
  Choose $L = \floor{(1-\gamma_{p,\varepsilon})/\eps}$.
  We show that, if $\eps$ is sufficiently small, $\mathcal{C}$ is with high probability not $(p,L-1)$-list decodable.

  Let $\mu=2^{-n}\Vol(n,pn)$.
  For any list decoding center $a\in\mathbb{F}_2^n$ and any ordered list of $L$ distinct messages $X=(x_1,\dots,x_L) \in (\mathbb{F}_2^{Rn})^L$, let $\mathbb{I}(a,X)$ be the indicator random variable for the event that the encodings of $x_1,\dots,x_L$ all fall in $\mathcal{B}(a,pn)$.
  Let $W = \sum_{a,X}\mathbb{I}(a,X)$.
  Note that $\mathcal{C}$ is $(p,L-1)$-list decodable if and only if $W=0$.

  Just as in \cite{GuruswamiN14}, we have $\E[\mathbb{I}(a,X)] = \mu^L$ and the number of different pairs $(a,X)$ is at least $\half2^n\cdot 2^{RnL}$. 
  Thus, we can compute $\E[W] \ge \half \mu^L2^n2^{RnL}$, so
  \begin{align}
    \Var W 
    \ &= \   \sum_{X,Y}\sum_{a,b}\left( \E_{\mathcal{C}}\left[ \mathbb{I}(a,X)\mathbb{I}(b,Y) \right] - \E_{\mathcal{C}}[\mathbb{I}(a,X)]\cdot\E_{\mathcal{C}}[\mathbb{I}(b,Y)] \right) \nonumber\\
    \ &= \   \sum_{X\cap Y\neq\emptyset}\sum_{a,b}\left( \E_{\mathcal{C}}\left[ \mathbb{I}(a,X)\mathbb{I}(b,Y) \right] - \E_{\mathcal{C}}[\mathbb{I}(a,X)]\cdot\E_{\mathcal{C}}[\mathbb{I}(b,Y)] \right) \nonumber\\
    \ &\le \   \sum_{X\cap Y\neq\emptyset}\sum_{a,b}\E_{\mathcal{C}}\left[ \mathbb{I}(a,X)\mathbb{I}(b,Y) \right] \nonumber\\
    \ &= \   \sum_{\ell=1}^L\sum_{|X\cap Y|=\ell}\sum_{a,b}\E_{\mathcal{C}}\left[ \mathbb{I}(a,X) \mathbb{I}(b,Y) \right] \nonumber\\
    \ &= \   \sum_{\ell=1}^L\sum_{|X\cap Y|=\ell}2^{2n}\Pr_{a,b,\mathcal{C}}\left[ \mathbb{I}(a,X) \text{ and }\mathbb{I}(b,Y) \right].
  \label{}
  \end{align}
  Suppose $X$ and $Y$ are $L$-tuples such that $|X\cap Y|=\ell$, where the intersection treats $X$ and $Y$ as sets. 
  Suppose that the elements of $X$ are $x_1,\dots,x_L$ and the elements of $Y$ are $x_1,\dots,x_{\ell}, y_{\ell+1},\dots,y_L$.
  In this notation, the event ``$\mathbb{I}(a,X)$ and $\mathbb{I}(b,Y)$'' is exactly event $\calE$ in Lemma~\ref{lem:lb-int}.
  Hence,
  \begin{align}
    \Var W 
    \ &\le \   \sum_{\ell=1}^L\sum_{|X\cap Y|=\ell}2^{2n}\Pr_{a,b,\mathcal{C}}\left[ \mathbb{I}(a,X) \text{ and }\mathbb{I}(b,Y) \right] \nonumber\\
    \ &\le \   \sum_{\ell=1}^L\sum_{|X\cap Y|=\ell}2^{2n}\cdot \min\left(\mu^{2L-\ell+1}, 2^{-n}\cdot \mu^{2L-\ell}\cdot C_p^L\cdot n^{L/2}\cdot (1 + 2^{-\alpha_{p}\ell})^n\right).
  \label{}
  \end{align}
  As $\alpha_{p} > 0$, by choice of $\gamma_{p,\varepsilon}$, for all $\ell \ge \ell_{p,\varepsilon}$, we have
  \begin{align*}
    \frac{1 + 2^{-\alpha_{p}\ell}}{2} \ \le \ 2^{-(1-2\gamma_{p,\varepsilon})}.
  \label{}
  \end{align*}
  We also bound the number of $X,Y$ such that $|X\cap Y|=\ell$ by $L^{2L}2^{Rn(2L-\ell)}$.
Indeed, there are at most $2^{Rn(2L - \ell)}$ ways to choose $X,Y$ as sets and at most $L^{2L}$ ways to order the sets.
Thus,
  \begin{align*}
    \frac{\Var W}{\E[W]^2}
    \ &\le \   
    \frac{4}{\mu^{2L}2^{2n}2^{2RnL}}
    \sum_{\ell=1}^L\sum_{|X\cap Y|=\ell}2^{2n}\cdot \min\left(\mu^{2L-\ell+1}, 2^{-n}\cdot \mu^{2L-\ell}\cdot C_p^L\cdot n^{L/2}\cdot (1 + 2^{-\alpha_{p}\ell})^n\right) \nonumber\\
    \ &\le \  
    \frac{4}{\mu^{2L}2^{2n}2^{2RnL}}
    \sum_{\ell=1}^{\ell_{p,\varepsilon}-1}\sum_{|X\cap Y|=\ell}2^{2n}\cdot \mu^{2L-\ell+1} \nonumber\\
    \ &\qquad \ + 
    \frac{4}{\mu^{2L}2^{2n}2^{2RnL}}
    \sum_{\ell=\ell_{p,\varepsilon}}^L\sum_{|X\cap Y|=\ell}2^{2n}\cdot \mu^{2L-\ell}\cdot C_p^L\cdot n^{L/2}\cdot \left(\frac{1 + 2^{-\alpha_{p}\ell}}{2}\right)^n \nonumber\\
    \ &\le \  
    \frac{4}{\mu^{2L}2^{2n}2^{2RnL}}
    \sum_{\ell=1}^{\ell_{p,\varepsilon}-1}L^{2L}2^{Rn(2L-\ell)}\cdot 2^{2n}\cdot \mu^{2L-\ell+1} \nonumber\\
    \ &\qquad \ + 
    \frac{4}{\mu^{2L}2^{2n}2^{2RnL}}
    \sum_{\ell=\ell_{p,\varepsilon}}^LL^{2L}2^{Rn(2L-\ell)}\cdot 2^{2n}\cdot \mu^{2L-\ell}\cdot C_p^L\cdot n^{L/2}\cdot \left(\frac{1 + 2^{-\alpha_{p}\ell}}{2}\right)^n \nonumber\\
    \ &\le \  
    \sum_{\ell=1}^{\ell_{p,\varepsilon}-1}4 L^{2L} (2^{Rn}\mu)^{-\ell}\cdot \mu
    + \sum_{\ell=\ell_{p,\varepsilon}}^L 4L^{2L}(2^{Rn}\mu)^{-\ell}\cdot C_p^L\cdot n^{L/2}\cdot 2^{-(1-2\gamma_{p,\varepsilon})n} \nonumber\\
    \ &\le \  4L^{2L+1}\cdot 2^{\eps\ell_{p,\varepsilon} n} \cdot 2^{-(1-H(p))n}
    + 4L^{2L+1}\cdot C_p^L\cdot n^{L/2}\cdot 2^{\eps nL - (1-2\gamma_{p,\varepsilon})n}
  \end{align*}
  By choice of $\ell_{p,\varepsilon}=\frac{1-H(p)}{2\varepsilon}$, the first term in the last line is $\exp(-\Omega_{p,\varepsilon}(n))$, and by choice of $L = \floor{\frac{1-\gamma_{p,\varepsilon}}{\eps}}$, the second term in the last line is also $\exp(-\Omega_{p,\varepsilon}(n))$. 
  Invoking Fact~\ref{fact:secondmoment},
  \begin{align*}
    \Pr[W=0] 
    \ \le \ \frac{\Var W}{\E[W]^2} 
    \ \le \ \exp(-\Omega_{p,\varepsilon}(n)).
  \label{}
  \end{align*}
Since $\mathcal{C}$ is $(p,L-1)$-list-decodable if and only if $W=0$, we conclude, for $n$ sufficiently large, that $\mathcal{C}$ is with probability $1-\exp(-\Omega_{p,\varepsilon}(n))$ not $(p,L-1)$-list decodable.
\end{proof}

\section{Lower bound on the list size of random rank metric codes}\label{app:rank-metric-lb}
In this section, we prove Theorem~\ref{thm:rank-metric-lb}.
The second moment method calculations in this section are not as delicate as in Section~\ref{app:lb}, because in Section~\ref{app:lb}, we wanted to pin down the list size to $(1\pm o(1))/\varepsilon$, but here showing list size $\Omega(1/\varepsilon)$ suffices.
\begin{theorem}[Theorem~\ref{thm:rank-metric-lb} restated]
  \thmrankmetriclbq
\end{theorem}

The proof is a second moment argument and is nearly identical to the lower bound for random linear codes over the Hamming metric in \cite{GuruswamiN14}: as this result is for the rank metric, we work out the details for completeness.
\begin{proof}
  Let $\eps > 0$ and fix $L = (1-p)(1-bp)/\eps - 1$.
Let $\mathcal{C} \subseteq \mathbb{F}_q^{m\times n}$ be a uniformly random rank metric code.  
We think of each $c \in \mathcal{C}$ as the encoding of a message $x \in \mathbb{F}_q^{Rmn}$.
  For any center $A \in \mathbb{F}_q^{m \times n}$ and any ordered list of $L$ distinct messages $X \defeq(x_1,x_2,\dots,x_L)\in (\mathbb{F}_q^{Rmn})^L$, let $\mathbb{I}(A,X)$ be the indicator random variable for the event that, for all $i$, the encodings of $x_1,\dots, x_L$ are all contained in $\mathcal{B}_R(A, pn)$.
  Define
  \begin{align*}
    W \ \defeq \ \sum_{A,X}\mathbb{I}(A,X).
  \label{}
  \end{align*}
  Note that $\mathcal{C}$ is $(p,L-1)$-list decodable if and only if $W = 0$.

  Let $\mu = q^{-mn}|\mathcal{B}_{q,R}(0,pn)|$. 
  By Lemma~\ref{lem:rank-metric-vol}, we have 
  \begin{align}
    \mu q^{Rmn} < 4q^{-\varepsilon mn}.
    \label{eq:rank-metric-vol-2}
  \end{align}
  The event that the encoding of $x_i$ falls inside $\mathcal{B}_{q,R}(A,pn)$ is exactly $\mu$, and these events are statistically independent, so we have $\E[\mathbb{I}(a,X)] = \mu^L$.
  Assuming $k\ge L+1$, for large enough $n$, the number of pairs $(A,X)$ is at least $\half q^{RmnL}\cdot q^{mn}$, as the number of ordered $L$-lists $X$ of $L$ distinct messages is
  \begin{align*}
    q^{Rmn}\cdot(q^{Rmn}-1)\cdots(q^{Rmn} - L+1)
    \ \ge \ q^{RmnL}\left( 1 - \frac{\binom{L}{2}}{q^{Rmn}} \right)
    \ \ge \ q^{RmnL}\left( 1 - \frac{q^L}{q^{Rmn}} \right)
    \ \ge \ \half q^{RmnL}.
  \label{}
  \end{align*}
  Thus, by linearity of expectation, we have
  \begin{align*}
    \E W \ \ge \ \half \mu^L q^{mn} q^{RmnL}.
  \label{}
  \end{align*}
  Following the method in \cite{GuruswamiN14}, we now bound the variance of $W$ above.
  For two lists of messages $X$ and $Y$, we let $|X\cap Y|$ denote the size of their intersection when we view them as sets.
  If $X$ and $Y$ are disjoint, then the events $\mathbb{I}(A,X),\mathbb{I}(B,Y)$ are independent for any pair of centers $A,B\in\mathbb{F}_q^{m\times n}$.
  Thus,
  \begin{align*}
    \Var W 
    \ &= \   \sum_{X,Y}\sum_{A,B}\left( \E_{\mathcal{C}}\left[ \mathbb{I}(A,X)\mathbb{I}(B,Y) \right] - \E_{\mathcal{C}}[\mathbb{I}(A,X)]\cdot\E_{\mathcal{C}}[\mathbb{I}(B,Y)] \right) \nonumber\\
    \ &= \   \sum_{X\cap Y\neq\emptyset}\sum_{A,B}\left( \E_{\mathcal{C}}\left[ \mathbb{I}(A,X)\mathbb{I}(B,Y) \right] - \E_{\mathcal{C}}[\mathbb{I}(A,X)]\cdot\E_{\mathcal{C}}[\mathbb{I}(B,Y)] \right) \nonumber\\
    \ &\le \   \sum_{X\cap Y\neq\emptyset}\sum_{A,B}\E_{\mathcal{C}}\left[ \mathbb{I}(A,X)\mathbb{I}(B,Y) \right] \nonumber\\
    \ &= \   \sum_{\ell=1}^L\sum_{|X\cap Y|=\ell}\sum_{A,B}\E_{\mathcal{C}}\left[ \mathbb{I}(A,X) \mathbb{I}(B,Y) \right] \nonumber\\
    \ &= \   \sum_{\ell=1}^L\sum_{|X\cap Y|=\ell}q^{2mn}\Pr_{A,B,\mathcal{C}}\left[ \mathbb{I}(A,X) \text{ and }\mathbb{I}(B,Y) \right]
  \label{}
  \end{align*}
  where, in the last line, $A,B$ are picked uniformly at random from $\mathbb{F}_q^{m\times n}$.

  Fix $0< l \le L$ and a pair $(X,Y)$ of $L$-tuples of matrices such that $|X\cap Y|=\ell$, and fix an arbitrary $z\in X\cap Y$, which is guaranteed to exist as $X,Y$ intersect.
  For the event $\mathbb{I}(A,X) = \mathbb{I}(B,Y) = 1$ to happen, all of the following must occur:
  \begin{enumerate}
  \item $A$ and $B$ are both within rank-metric distance $pn$ of the encoding of $z$ .
  \item For each $x\in X\setminus \{z\}$, the encoding of $x$ falls inside $B_{q,R}(A,pn)$.
  \item For each $y\in Y\setminus X$, the encoding of $y$ falls inside $B_{q,R}(B, pn)$.
  \end{enumerate}
  The first event occurs with probability $\mu^2$, and conditioned on the choices of $A$ and $B$, the second and third events occur with probabilities $\mu^{L-1}$ and $\mu^{L-\ell}$, respectively, and are independent given $A$ and $B$.
  Thus, the probability all the events occur is $\mu^{2L-\ell + 1}$, so the probability that $\mathbb{I}(A,X) = \mathbb{I}(B,Y) = 1$ is at most $\mu^{2L-\ell+1}$.
  Finally, noting that the number of pairs $(X,Y)$ with $|X\cap Y| = L$ is bounded by $L^{2L}\cdot q^{Rmn(2L-\ell)}$, we conclude
  \begin{align*}
    \Var W
    \ \le \ \sum_{\ell=1}^L L^{2L}\cdot q^{Rmn(2L-\ell)} \cdot q^{2mn}\mu^{2L-\ell+1}.
  \label{}
  \end{align*}
  Using Fact~\ref{fact:secondmoment} along with \eqref{eq:rank-metric-vol-2}, we have
  \begin{align*}
    \Pr[W = 0]
    \ &\le \ \frac{\Var W}{\E[W]^2}
    \ \le \   \sum_{\ell=1}^L 4L^{2L}(q^{Rmn}\mu)^{-\ell}\mu
    \ \le \ 4L^{2L+1}\cdot 4^Lq^{\eps mnL}\cdot 4q^{-mn(1-p)(1-bp)}.
  \label{}
  \end{align*}
  This quantity is $\exp_{q}(-\Omega_{p,\eps}(n))$ for our choice of $L$.  Recalling that $\mathcal{C}$ is $(p,L-1)$-list-decodable if and only of $W = 0$, we conclude that with high probability, $\mathcal{C}$ is \em not \em $(p,L-1)$-list-decodable, as desired.
\end{proof}

\section{Existence of rate $1-H(p)-\eps$ codes achieving list size $\frac{H(p)}{\eps}$}\label{app:lll}
In this section, we use the Symmetric Lovasz Local Lemma (see, e.g., \cite{AlonS92}).
\begin{lemma}[Lovasz Local Lemma]
  \label{lem:lll}
  Let $d$ be a positive integer and $p\in(0,1)$.
  Let $A_1,\dots, A_N$ be events and $G=(V,E)$ be a graph with vertices $V=[N]$ such that each vertex has degree at most $d$, and, for all $i\in[N]$, the event $A_i$ is independent of the events $\{A_j: j\in [N], j\neq i, ij\notin E\}$.
  Suppose also that for all $i$, $\Pr[A_i] \le p$.
  If $ep(d+1) < 1$, then 
  \begin{align}
    \Pr\left[\bigwedge_{i=1}^N\overline{A_i}\right] \ge \left(1-\frac{1}{d+1} \right)^N > 0,
  \end{align}
  i.e. with positive probability, none of the events $A_i$ hold.
\end{lemma}

We prove the following.
\begin{theorem}
  \label{thm:lll}
  Let $p\in[0,1/2)$ and $\eps > 0$.
  For any large enough $n$, there exist binary codes of rate $R = 1 - H(p) - \eps$ that are $(p,\frac{H(p)}{\eps}+1)$-list decodable.
\end{theorem}
\begin{remark}
  Theorem~\ref{thm:lll} does not contradict Theorem~\ref{thm:lb}, our high probability lower bound, as Theorem~\ref{thm:lll}  gives an existential result, rather than a high probability result.
\end{remark}
\begin{proof}
  Let $M = 2^{Rn+1}$.
  Let $L = \floor{H(p)/\eps} + 1 > \frac{H(p)}{\eps}$.
  Take a random code $\mathcal{C}$ where $c_1,\dots, c_M$ are chosen from $\mathbb{F}_2^n$ independently and uniformly at random.
  For $y\in\mathbb{F}_2^n$ and $I=\{i_1,\dots, i_L\}$, let $A(y, I)$ denote the event that $c_{i_j}\in \mathcal{B}(y,pn)$ for $j=0,\dots,L$.
  There are $2^nM^{L+1}$ such events.
  We have that $\mathcal{C}$ is $(p, L)$-list decodable if and only if none of the events $A(y, I)$ occur.
  
  Define a graph $G$ on the events $A(y, I)$ where $A(y,I), A(z,J)$ are connected by an edge if $I\cap J\neq\emptyset$.
  Note that, as the $c_i$'s are all independent and each $A(y,I)$ depends only on $c_i$ such that $i\in I$, we have $A(y,I)$ is independent of all $A(z,J)$ satisfying $I\cap J=\emptyset$. 
  In other words, $A(y,I)$ is independent from all nonneighbor events in $G$.
  For each $I\subseteq[M]$ with $|I|=L+1$, there are at most $(L+1)\cdot M^L$ choices of $J$ such that $I\cap J\neq \emptyset$.
  Thus, each $A(y,I)$ has degree at most $d \defeq 2^n\cdot (L+1)\cdot M^L$.
  On the other hand, we can compute 
  \begin{align}
    p\ &\defeq \  \Pr[A(y,I)] \ = \   \left( \frac{\Vol(n,pn)}{2^n} \right)^{L+1} \ \le \ 2^{-n(L+1)(1-H(p))}.
  \end{align}
  We thus have
  \begin{align}
    e\cdot p(d+1) 
    \ &\le \   3\cdot 2^{-n(L+1)(1-H(p))}\cdot (L+1)\cdot 2^n\cdot 2^{RnL+L} \nonumber\\
    \ &< \   6L\cdot 2^{-n(L+1)(1-H(p))}\cdot 2^n\cdot 2^{(1-H(p)-\eps)nL+L} \nonumber\\
    \ &= \   6L\cdot 2^{nH(p)-\eps nL+L}  < 1,
  \label{}
  \end{align}
  for our choice of $L$ and sufficiently large $n$.
  Thus, by Lemma~\ref{lem:lll}, we have
  \begin{align}
    \Pr\left[ \bigwedge_{y,I}\overline{A(y,I)} \right] 
    \ &\ge \ \left(1-\frac{1}{d+1} \right)^{2^nM^{L+1}} \nonumber\\
    \ &> \ \exp\left(-\frac{2^nM^{L+1}}{d} \right)    \nonumber\\
    \ &= \ \exp\left(-\frac{M}{L+1} \right) \ > \ 0. 
  \label{}
  \end{align}
  For the second inequality, we used the estimate $1-\frac{1}{d+1}\ge e^{-\frac{1}{d}}$. %Ray: this is proved by taking reciprocal
  At this point, we are essentially done, but it is possible in our random choice that, for example, all $c_i$ are the same point.
  We check that this probability is sufficiently small compared to the probability of list decodability.

  There are $2^{nM}$ ways to choose $c_1,\dots, c_M$.
  The probability that $|\mathcal{C}|\le M/2$ is at most,
  \begin{align}
    \Pr\left[ |\mathcal{C}|\le M/2 \right]
    \ &\le \ \frac{\binom{2^n}{M/2}\cdot (M/2)^M}{2^{nM}}
    \ \le \ \frac{\frac{2^{nM/2}}{(M/2)^{M/2}e^{-M/2}}\cdot (M/2)^M}{2^{nM}} \nonumber\\
    \ &= \ \frac{e^{M/2}(M/2)^{M/2}}{2^{nM/2}} 
    \ = \ \left( \sqrt{\frac{eM/2}{2^n}} \right)^M \nonumber\\
    \ &< \ e^{-\frac{M}{L+1}}
    \ < \  \Pr\left[ \bigwedge_{y,I}\overline{A(y,I)} \right] 
  \end{align}
  where the second inequality was bounded by Stirling's approximation.
  
  Thus, with positive probability, we have both (1) none of the events $A(y,I)$ hold, i.e. $\mathcal{C}$ is $(p,L)$-list decodable, and (2) $\mathcal{C}$ has rate at least $R=1-H(p)-\eps$, so in particular there exists some $\mathcal{C}$ with the desired list decodability and rate.
\end{proof}

\end{document}